%% file: main.tex
\documentclass{article}
\def\BibTeX{{\rm B\kern-.05em{\sc i\kern-.025em b}\kern-.08em
    T\kern-.1667em\lower.7ex\hbox{E}\kern-.125emX}} 

\usepackage[margin=1in]{geometry}
\usepackage{amsmath}
\usepackage{amsfonts,amssymb}
\usepackage{xr-hyper}
\usepackage[hidelinks]{hyperref}

\input{preamble}

\usepackage{bigfoot}

\DeclareNewFootnote{AAffil}[arabic]
\DeclareNewFootnote{ANote}[fnsymbol]

\usepackage{etoolbox}
\makeatletter
\patchcmd\maketitle{\def\@makefnmark{\rlap{\@textsuperscript{\normalfont\@thefnmark}}}}{}{}{}
\makeatother

\makeatletter
\def\thanksAAffil#1{% <--- These %'s are necessary for spacing
  \footnotemarkAAffil\protected@xdef\@thanks{\@thanks%
        \protect\footnotetextAAffil[\the \c@footnoteAAffil]{#1}}%
}
\def\thanksANote#1{%
  \footnotemarkANote%
  \protected@xdef\@thanks{\@thanks%
        \protect\footnotetextANote[\the \c@footnoteANote]{#1}}%
}
\makeatother

\author{% <---- Not sure if these %'s are necessary, but can't hurt
  Amit Surana%
  \thanks{amit.surana@rtx.com, RTX Technology Research Center (RTRC), East Hartford, CT, USA.}%
  $^{\hspace{5pt} ,}$\footnotemark[3]
  \hspace{5pt}, %
  Abeynaya Gnanasekaran%
  \thanks{abeynaya.gnanasekaran@rtx.com, RTX Technology Research Center (RTRC), Berkeley, CA, USA.}%
  $^{\hspace{5pt} ,}$\thanks{Both authors contributed equally to this work.}%
}

\title{Variational Quantum Framework for Partial Differential Equation Constrained Optimization}

\begin{document}
\maketitle

\begin{abstract}
We present a novel variational quantum framework for linear partial differential equation (PDE) constrained optimization problems. Such problems arise in many scientific and engineering domains. For instance, in aerodynamics, the PDE constraints are the conservation laws such as momentum, mass and energy balance, the design variables are vehicle shape parameters and material properties, and the objective could be to minimize the effect of transient heat loads on the vehicle or to maximize the lift-to-drag ratio. The proposed framework utilizes the variational quantum linear system (VQLS) algorithm and a black box optimizer as its two main building blocks. VQLS is used to solve the linear system, arising from the discretization of the PDE constraints for given design parameters, and evaluate the design cost/objective function. The black box optimizer is used to select next set of parameter values based on this evaluated cost, leading to nested bi-level optimization structure within a hybrid classical-quantum setting. We present detailed computational error and complexity analysis to highlight the potential advantages of our proposed framework over classical techniques. We implement our framework using the PennyLane library, apply it to a heat transfer optimization problem, and present simulation results using Bayesian optimization as the black box optimizer.
\end{abstract}

\section{Introduction}
We present a novel variational quantum framework for optimization problems constrained by linear partial differential equations (PDE). Variational quantum algorithms (VQAs) are hybrid classical-quantum algorithms that have emerged as promising candidates to optimally utilize today's Noisy Intermediate Scale Quantum (NISQ) devices. VQAs use a quantum subroutine to evaluate the cost/ energy function using a shallow parameterized quantum circuit. They provide a versatile framework to tackle a wide variety of problems. As such, they have been successfully used in many applications such as finding ground and excited states of Hamiltonians~\cite{app1,app2,app3,app4}, dynamical quantum simulation~\cite{ds1,ds2}, combinatorial optimization~\cite{bp_an1,c2}, solving linear systems~\cite{VQLS,ls2,ls3}, integer factorization~\cite{f1}, principal component analysis~\cite{pca1,pca2,pca3}  and quantum machine learning~\cite{biamonte2017quantum}. 

In this work we are interested in the application of the variational quantum linear solver (VQLS)~\cite{VQLS}, a type of VQA for solving ordinary differential equations (ODEs) and partial differential equations (PDEs)~\cite{liu2022application,demirdjian2022variational,liu2021variational,leong2022variational}. In this framework the ODEs/PDEs are discretized in space/time and embedded into a system of linear equations. VQLS is then applied to the linear system to obtain the normalized solution and compute observables of interest. While there are no theoretical guarantees associated with VQLS, its query complexity is empirically found to scale polylogarithmically with the size of the linear system~\cite{VQLS} and thus could be beneficial over classical linear solvers that have a polynomial scaling. Unlike variational algorithms, fault tolerant quantum algorithms (FTQA) for ODE/PDE simulation have also been proposed and have provable theoretical guarantees/advantages over classical methods. Such FTQA rely either on quantum linear system solver or Hamiltonian simulation framework, see for example~\cite{
berry2017quantum,li2023potential,babbush2023exponential,jennings2023cost,childs2021high,liu2021efficient,jin2022quantum,surana2024efficient,gnanasekaran2023efficient}. While, such FTQA can be used in our proposed PDE constrained optimization framework, in this paper we restrict to a variational setting with a focus on NISQ-era implementation. In particular, we extend the VQLS based ODE/PDE solution framework for linear PDE constrained design optimization problems. 

Many science and engineering applications necessitate the solution of optimization problems constrained by physical laws that are described by systems of PDEs. Such problems arise in several domains including aerodynamics, computational fluid dynamics (CFD), material science, computer vision, and inverse problems~\cite{antil2018frontiers}. For instance in aerodynamic design the PDE constraints are the conservation laws for momentum, mass,  and energy, the design variables are aerodynamic shape parameters and material properties, and the objective is to minimize heat transfer or to maximize lift-to-drag ratio. Inverse problems on other hand typically involve calibrating PDE model parameters to match the given measurements. Closed-form solutions are generally unavailable for PDE constrained optimization problems, necessitating the development of numerical methods~\cite{biegler2003large,hinze2008optimization,antil2018frontiers}. A variety of classical gradient-based and gradient-free numerical methods have been developed, and rely on repeated calls to the underlying PDE solver. Since PDE simulations are computationally expensive, using them within the design/optimization loop can become a bottleneck. 

To address this challenge we propose a variational quantum framework where a quantum device is used for the computationally demanding PDE simulations and design/optimization is carried out on a classical computer. To best of our knowledge, this is the first work considering the simulation and design problem jointly in the context of quantum algorithms. The main contributions of this work are as follows:
\begin{itemize}
  \item We formulate a bi-level variational quantum framework for solving linear PDE constrained optimization problems as discussed above. We refer to this as bi-level variational quantum PDE constrained optimization (BVQPCO) framework. Our framework utilizes the VQLS algorithm and a black box classical optimizer as its two main building blocks. VQLS is used to solve the linear system constraints, arising from the discretization of the underlying PDEs, for given design parameters, and evaluate the cost/objective function, while the black box optimizer is used to select the next set of parameter values based on the evaluated cost. Since, VQLS itself involves optimization of the ansatz parameters which parameterize the solution of the linear system, this leads to a sequence of nested bi-level optimization steps within a VQA setting. 
  \item We present detailed computational error analysis for VQLS based solution of linear ODEs under explicit and implicit Euler discretization schemes. Furthermore by combining this error analysis with the empirically known results for query complexity of VQLS, we assess potential utility of our BVQPCO framework over the classical methods.  
  \item We implement the BVQPCO framework using the PennyLane library, apply it to solve a prototypical heat transfer optimization problem, and present simulation results using Bayesian optimization (BO) as the black box optimizer.
\end{itemize}

The paper is organized as follows. In \Cref{sec:PDEopt} we motivate the general structure of PDE constrained design optimization problem through a detailed mathematical formulation of the 1D heat transfer optimization problem. We describe our BVQPCO framework and outline a pseudo-code for its implementation in~\Cref{sec:bilevel}. \Cref{sec:VQLS,sec:BO} provide an overview of VQLS and BO, respectively.  We  present the main results from the complexity analysis of the BVQPCO framework in~\Cref{sec:erranalysis}; detailed proofs are provided in the Appendix. In ~\Cref{sec:num} we illustrate the BVQPCO framework on the heat transfer optimization problem and provide details on the implementation using the PennyLane library. We finally conclude in~\Cref{sec:conc} with avenues for future work.  

\section{Notation}
Let $\Nr = \{1, 2,\dots\}$, $\Rr$, and $\Cr$ be the sets of positive integers,  real numbers, and  complex numbers respectively. We will denote vectors by small bold letters and  matrices by capital bold letters. $\Am^*$ and $\Am^{\prime}$ will denote the vector/matrix complex conjugate and vector/matrix transpose, respectively. $Tr(\Am)$ will denote the trace of a matrix. We will represent an identity matrix of size $s \times s$ by $\Id_s$.  Kronecker product will be denoted by $\otimes$. If  $\Am\in\Rr^{m\times n}$ and $\mathbf{B}\in\Rr^{p\times q}$, then their Kronecker product is given by
\begin{equation}\label{eq:kronM}
\mathbf{C}=\Am\otimes\mathbf{B}=\left(
                                 \begin{array}{ccc}
                                   a_{11}\mathbf{B} & \cdots & a_{1n}\mathbf{B} \\
                                   \vdots & \vdots & \vdots \\
                                    a_{m1}\mathbf{B} & \cdots & a_{mn}\mathbf{B} \\
                                 \end{array}
                               \right),
\end{equation}
where, $\mathbf{C}\in \Rr^{mp\times nq}$. 

The standard inner product between two vectors $\mathbf{x}$ and $\mathbf{y}$ will be denoted by $\la\mathbf{x},\mathbf{y}\ra$. The $l_p$-norm in the Euclidean space $\Rr^n$ will be denoted by $\|\cdot\|_p,p=1,2,\dots,\infty$ and is defined as follows
\begin{equation}\label{eq:normx}
\|\xv\|_p=\left(\sum_{j=1}^n|x_j|^p\right)^{1/p},\quad \xv\in \Rr^n.
\end{equation}
For the norm of a matrix $\Am \in \Rr^{n\times m}$, we use the induced norm, namely
\begin{equation}\label{eq:normA}
\|\Am\|_p=\max_{\xv \in \Rr^m} \frac{\|\Am \xv\|_p}{\|\xv\|_p}.
\end{equation}
We will use $p=2$, i.e., $l_2$ norm for vectors and spectral norm for matrices unless stated otherwise. The spectral radius of a matrix denoted by $\rho(\Am)$ is defined as 
\begin{equation}\label{eq:specA}
\rho(\Am)=\max\{|\lambda|:\lambda \mbox{ eignenvalues of } \Am\}.
\end{equation}
Trace norm of a matrix $\Am$ is defined as $\|\Am\|_{Tr}=Tr(\sqrt{\Am^*\Am})$ which is the  Schatten $q$-norms with $q=1$.

We will use standard braket notation, i.e. $|\psiv\ra$ and $\la\psiv|$ in representing the quantum states \cite{nielsen2010quantum}. The inner product between two quantum states $|\psiv\ra$ and $|\phiv\ra$ will be denoted by $\la\psiv|\phiv\ra$. For a vector $\mathbf{x}$, we denote by $|\mathbf{x}\ra=\frac{\mathbf{x}}{\|\mathbf{x}\|}$ as the corresponding quantum state. The trace norm $\rho(\psiv,\phiv)$ between two pure states $|\psiv\ra$ and $|\phiv\ra$ is defined as
\begin{equation}\label{eq:trace}
\rho(\psiv,\phiv)=\frac{1}{2}\| |\phiv\ra \la\phiv| -|\psiv\ra \la\psiv| \|_{Tr}=\sqrt{1-|\la \phiv|\psiv \ra|^2}.
\end{equation}

\section{PDE Constrained Optimization}\label{sec:PDEopt}
We consider a general class of PDE constrained design optimization problems of the form
\begin{eqnarray}
&&\min_{\mathbf{p},\mathbf{u}} C_d(\mathbf{u}, \mathbf{p}) \label{eq: gen_opt}\\
\text{subject to,} &&\mathbf{F}(\mathbf{u},\mathbf{p},t) = 0,\label{eq:pdecons}\\
&& g_i(\mathbf{p})\leq 0, i=1,\cdots,N_e,
\end{eqnarray}
where, $\mathbf{p}\in\Rr^{N_p}$ is the vector of design variable (e.g., material type, shape), $C_d$ is the cost function (e.g., heat flux, drag/ lift) and $\mathbf{F}(\mathbf{u}, \mathbf{p},t)$ are the PDE constraints (e.g., conservation laws, constitutive relations, boundary/ initial conditions) with $\mathbf{u}(\mathbf{x},t;\pv)$ being the solution of the PDE defined as a function of space $\mathbf{x}$ and time $t$ for given parameters, and $g_i(\mathbf{p}),i=1,\cdots,N_e$ are constraints on the parameters~\cite{de2015numerical}. Note that $\mathbf{F}$ depends on $\mathbf{u}$ and its partial derivatives in space/time, which we have not explicitly represented in $\mathbf{F}$ for brevity. 

In this paper we will restrict $\mathbf{F}$ to arise from linear PDEs. The cost function $C_d$ can in general be non-linear, non-convex and non-smooth. The design variables $\mathbf{p}$ can be mixed, i.e. both continuous and discrete valued. Though the underlying PDE is linear, the constraints $\mathbf{F}(\mathbf{u},\mathbf{p},t)$ could involve products of $\mathbf{u}$ and $\mathbf{p}$ making them nonlinear. For numerical purposes, the PDE is discretized in space and/or time and the discretized solution $\mathbf{u}\in \Rr^{\nx}$ (with slight abuse of notation) will be typically a very high dimensional vector, i.e., $\nx\gg 1$. Thus, in general the optimization problem is computationally challenging to solve. 

One popular classical approach to solve PDE constrained optimization is via gradient based methods~\cite{biegler2003large,hinze2008optimization,antil2018frontiers}. For example, the gradient of the cost function $C_d$ w.r.t to $\pv$ can be obtained via finite differencing. However, that requires solving the PDE (\ref{eq:pdecons}) at least $N_p+1$ times to approximate the derivative of the solution $\mathbf{u}$ with respect to the parameters. A more efficient approach is via adjoint formulation, where one needs to solve the PDE (\ref{eq:pdecons}) only one time and an additional adjoint equation to compute the required gradient \cite{biegler2003large,hinze2008optimization}. However, given that this gradient based approach requires forward simulation(s) for each gradient step, it can be computationally intensive. An alternative approach is surrogate model assisted optimization \cite{jiang2020surrogate,audet2000surrogate}, which has particularly gained significant interest with the recent advances in machine learning \cite{raissi2019physics,cai2021physics,sun2023physics,lu2021physics}.  In this approach a surrogate model is trained on the PDE variables calculated with solutions from a PDE solver at several different parameter values. The surrogate model which is computationally cheaper to evaluate than the PDE solver, can then be used for optimization. However, this requires large amount of training data or PDE solves, and the model is only approximate. Another related approach is black box optimization, such as Bayesian optimization (see Section \ref{sec:BO}) where the PDE solver is treated as a black box and surrogate model construction and optimization are accomplished simultaneously by sequentially/adaptively sampling parameters for querying the PDE solver \cite{morita2022applying,sarkar2019multifidelity,diessner2022investigating}. 

In all the above mentioned approaches obtaining the PDE solution repeatedly for different parameter values is the computational bottleneck. To alleviate this challenge, we propose to use VQLS based PDE solution, and embed that within a classical black box optimization scheme for design/optimization as discussed above.

\subsection{Heat Transfer Optimization Problem Formulation}\label{sec:heatqen}
In high temperature aerodynamic applications such as space reentry vehicles, specially-designed heat-resistant coating layers or materials are applied to protect the vehicle body from excessive heat loading. In such applications it is critical to choose appropriate thickness and material property (thermal diffusivity) of the coating layer. While a thicker coating layer improves the level of insulation, it also makes it more susceptible to thermo-mechanical cracking. Thus, one needs to choose the appropriate material while meeting the thickness constraints. To model this problem we consider a canonical setup of a uniform coating layer applied on a flat plate. The plate is subject to a time-dependent heat flux on one end as shown in Fig.~\ref{fig: flat_plate}.

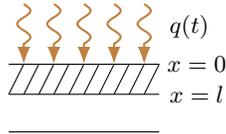
\begin{figure}[!htbp]
\centering
\input{heat_flow.tikz}
\caption{Design of coating layers for a flat plate subject to a time-dependent flux.}
\label{fig: flat_plate}
\end{figure}

To model the heat flow through the plate, we use the 1D heat equation over the domain $[0,l]$ with Neumann boundary conditions
\begin{eqnarray}
  \frac{\partial T(x,t)}{\partial t}&=&\alpha \frac{\partial^2 T(x,t)}{\partial x^2},\label{eq:heateqn}\\
  -k\frac{\partial T(x,t)}{\partial x}|_{x=0} &=& q(t), \label{eq:bc1}\\
  -k\frac{\partial T(x,t)}{\partial x}|_{x=l} &=& 0,\label{eq:bc2}\\
   T(x,0)&=&T_0(x),\label{eq:ic}
\end{eqnarray}
where, $t\in [0,P]$, $q(t)$ is time dependent heat flux, $k$ is thermal conductivity, and $\alpha$ is thermal diffusivity. We consider following optimization problem
\begin{eqnarray}
&& \min_{l,\alpha}C_d(l,\alpha)\equiv \frac{w_1}{P}\frac{\int_{t=0}^P T(l,t)dt}{T_0(l)}+w_2\left(\frac{l}{l_{\min}}-1\right)^2+w_3\left(\frac{\alpha}{\alpha_{\max}}-1\right)^2, \label{eq:pdco}\\
\text{subject to:}&& \text{PDE constraints (\ref{eq:heateqn}-\ref{eq:ic})},\notag\\
&& l_{\min}\leq l\leq l_{\max},\quad \alpha_{\min} \leq \alpha\leq \alpha_{\max},\notag
\end{eqnarray}
where, $w_i\geq0,i=1,2,3$ are user selected weights. The optimization problem trades off the thickness $l$ and $\alpha$ with the average temperature at the interface $x=l$.

Discretizing the PDE (\ref{eq:heateqn}) with second order accuracy in space and the boundary conditions (\ref{eq:bc1}) and (\ref{eq:bc2}) with first order accuracy in space, the vector $\Tv(t)=(\,T_1(t), \dots, T_N(t)\,)^{\prime}$, where $T_i(t)=T(x_i,t), i=1, \dots, N$, $x_i=(i-1)\Delta x$ and $\Delta x =\frac{l}{N-1}$,  satisfies an ordinary differential equation (ODE)
\begin{equation*}
\frac{d T_i(t)}{d t}=\frac{\alpha}{(\Delta x)^2} (T_{i+1}-2T_i+T_{i-1}),
\end{equation*}
with
\begin{eqnarray}
  -k\frac{T_{1}(t)-T_{0}(t)}{\Delta x}&=& q(t), \Rightarrow  T_{0}(t)= \frac{q(t)\Delta x }{k}+T_1(t),\\
  -k\frac{T_{N+1}(t)-T_{N}(t)}{\Delta x}&=&0, \Rightarrow    T_{N+1}=T_{N}.
\end{eqnarray}
Therefore, for $i=1$
\begin{equation*}
\frac{d T_1(t)}{d t}=\frac{\alpha}{\Delta x^2} (T_{2}-2T_1+T_{0})=\frac{\alpha}{(\Delta x)^2} (T_2-T_1)+\frac{\alpha }{k\Delta x} q(t),
\end{equation*}
and, $i=N$
\begin{equation*}
\frac{d T_N}{d t}=\frac{\alpha}{(\Delta x)^2} (-T_N+T_{N-1}).
\end{equation*}
Thus one can express the above system of ODEs in a compact vector form as:
\begin{equation}
\frac{d \Tv}{d t}=\frac{\alpha}{(\Delta x)^2} \Am \Tv+\frac{q(t)}{k\Delta x}\ev_1, \label{eq:spaceODE}
\end{equation}
with the initial condition $\Tv(0)=\Tv_0=(\,T_0(x_1), \dots, T_0(x_N)\,)^\prime$. Here,  $\ev_i=(0, \dots, 1, \dots, 0)^\prime$ is the standard unit vector with entry $1$ at index $i$ and zero elsewhere and 
\begin{equation}\label{eq:A}
\Am=\begin{pmatrix*}[r]
 -1 & 1 &  &  & 0 \\
       1 & -2 & 1 &  &  \\
       & \ddots & \ddots & \ddots &  \\
       &  & \ddots  & -2 & 1 \\
      0 &  &  & 1 & -1 \\
\end{pmatrix*}.
\end{equation}

\subsubsection{Explicit Time Discretization} \label{sec:expED}
Furthermore, discretizing the Eqn.~(\ref{eq:spaceODE}) in time using the explicit or forward Euler scheme leads to
\begin{equation*}
\Tv^{k+1}_f=\bigg(\In+\frac{\alpha \Delta t}{(\Delta x)^2} \Am\bigg) \Tv^k_f+q^k \frac{\Delta t}{k\Delta x}\ev_1,\quad k=1,\dots,M-1
\end{equation*}
where, $\Tv^{k}_f=(T_{1k},T_{2k},\dots,T_{Nk})^\prime$ with $T_{ik}=T(x_i,t_k)$, $q^k=q((k-1)\Delta t)$, $t_k=(k-1)\Delta t$ and $\Delta t=\frac{P}{M-1}$. Note that for this discretization scheme to be stable, CFL condition should be met $\frac{\alpha \Delta t}{(\Delta x)^2}\leq \frac{1}{2}$.

One can express the set of difference equations above in a form of a single system of linear algebraic equation
\begin{equation}\label{eq:linsysexp}
\tilde{\Am}_f\tilde{\Tv}_f=\tilde{\bv},
\end{equation}
where, $\tilde{\Tv}_f=(\Tv^1_f,\Tv^2_f, \dots,\Tv^M_f)^\prime$, $\tilde{\bv}=(\Tv_0,\frac{\Delta t}{k\Delta x}q^1\ev_1,\frac{\Delta t}{k\Delta x}q^2\ev_1, \dots,\frac{\Delta t}{k\Delta x}q^{M-1}\ev_1)^\prime$
and
\begin{eqnarray}
\tilde{\Am}_f&=&\left(
    \begin{array}{ccccc}
      \In & 0 & 0 & \cdots & 0 \\
      -\In-\frac{\alpha \Delta t}{(\Delta x)^2} \Am & \In & 0 & 0 & \cdots \\
      0 &  -\In-\frac{\alpha \Delta t}{(\Delta x)^2} \Am & \In & 0 & \cdots \\
      \vdots & \vdots & \vdots  & \cdots & \vdots \\
      0 & 0 & \cdots &  -\In-\frac{\alpha \Delta t}{(\Delta x)^2} \Am &  \In\\
    \end{array}
  \right)\notag\\
 &=&  \left(
    \begin{array}{ccccc}
      \In & 0 & 0 & \cdots & 0 \\
      -\In & \In & 0 & 0 & \cdots \\
      0 & -\In & \In & 0 & \cdots \\
      \vdots & \vdots & \vdots  & \cdots & \vdots \\
      0 & 0 & \cdots & -\In & \In \\
 \end{array} \right)-\frac{\alpha \Delta t}{(\Delta x)^2} \left(
    \begin{array}{ccccc}
      0 & 0 & 0 & \cdots & 0 \\
      \Am & 0 & 0 & 0 & \cdots \\
      0 &  \Am & 0 & 0 & \cdots \\
      \vdots & \vdots & \vdots  & \cdots & \vdots \\
      0 & 0 & \cdots & \Am & 0 \\
  \end{array}\right)\notag \\
 &=&\tilde{\Am}_{f1}-\frac{\alpha \Delta t}{(\Delta x)^2}\tilde{\Am}_{f2}. 
  \addtocounter{equation}{-1}\refstepcounter{equation}
  \label{eq:Aexp}
\end{eqnarray}

To introduce explicit dependence on $l$, we define a rescaling $x=ly$, so that $y\in [0,1]$ and $\Delta x= l \Delta y$, leading to:
\begin{eqnarray*}
\tilde{\bv}(l,\alpha)&=& (\Tv_0,\frac{\Delta t}{kl\Delta y}q^1\ev_1^\prime,\frac{\Delta t}{kl\Delta y}q^2\ev_1^\prime, \dots,\frac{\Delta t}{k l\Delta y}q^{M-1}\ev_1^\prime)^\prime,\\
\tilde{\Am}_f(l,\alpha)&=&\tilde{\Am}_{f1}-\frac{\alpha \Delta t}{l^2 (\Delta y)^2}\tilde{\Am}_{f2}.
\end{eqnarray*}

\subsubsection{Implicit Time Discretization}\label{sec:impED}
Implicit or backward Euler time discretization scheme is unconditionally stable, and leads to
\begin{equation*}
\frac{\Tv^{k+1}_b-\Tv^k_b}{\Delta t}=\frac{\alpha }{(\Delta x)^2} \Am \Tv^{k+1}_b+q^k \frac{1}{k\Delta x}\ev_1,
\end{equation*}
or equivalently
\begin{equation*}
\bigg(I-\frac{\alpha \Delta t}{(\Delta x)^2} \Am\bigg)\Tv^{k+1}_b= \Tv^k_b+q^k \frac{\Delta t}{k\Delta x}\ev_1.
\end{equation*}
As for the explicit case, we can express the set of difference equations above into a single linear system
\begin{equation}\label{eq:linsysimp}
\tilde{\Am}_b\tilde{\Tv}_b=\tilde{\bv},
\end{equation}
where, $\tilde{\Tv}_b$ and $\tilde{\bv}$ are similarly defined as above, and
\begin{eqnarray}
\tilde{\Am}_b&=&\left(
    \begin{array}{ccccc}
      \In & 0 & 0 & \cdots & 0 \\
      -\In & (\In-\frac{\alpha \Delta t}{(\Delta x)^2} \Am) & 0 & 0 & \cdots \\
      0 & -\In & (\In-\frac{\alpha \Delta t}{(\Delta x)^2} \Am)  & 0 & \cdots \\
      \vdots & \vdots & \vdots  & \cdots & \vdots \\
      0 & 0 & \cdots & -\In & (\In-\frac{\alpha \Delta t}{(\Delta x)^2} \Am)  \\
    \end{array}
  \right)\notag\\
 &=&  \left(
    \begin{array}{ccccc}
      \In & 0 & 0 & \cdots & 0 \\
      -\In & \In & 0 & 0 & \cdots \\
      0 & -\In & \In & 0 & \cdots \\
      \vdots & \vdots & \vdots  & \cdots & \vdots \\
      0 & 0 & \cdots & -\In & \In \\
 \end{array} \right)-\frac{\alpha \Delta t}{(\Delta x)^2} \left(
    \begin{array}{ccccc}
      0 & 0 & 0 & \cdots & 0 \\
      0 & \Am & 0 & 0 & \cdots \\
      0 &  0 & \Am & 0 & \cdots \\
      \vdots & \vdots & \vdots  & \cdots & \vdots \\
      0 & 0 & \cdots & 0 & \Am \\
  \end{array}\right)\notag \\
 &=&\tilde{\Am}_{b1}-\frac{\alpha \Delta t}{(\Delta x)^2}\tilde{\Am}_{b2}. 
 \addtocounter{equation}{-1}\refstepcounter{equation}
 \label{eq:Aim}
\end{eqnarray}
Similar to the explicit case, one can introduce rescaling with $l$ in the above expression.

\subsubsection{Optimization Problem}
With these discretizations and rescaling, we can finally express the optimization problem (\ref{eq:pdco}) as:
\begin{eqnarray}
&& \min_{\mathbf{p},\tilde{\Tv}_s} C_d(\mathbf{p},\tilde{\Tv}_s)\equiv f(\frac{\langle \phiv_d,\tilde{\Tv}_s\rangle}{\langle \phiv_n,\tilde{\Tv}_s\rangle},\mathbf{p})= w_1 \frac{\Delta t}{P}\frac{\langle \phiv_d,\tilde{\Tv}_s\rangle}{\langle \phiv_n,\tilde{\Tv}_s\rangle}+w_2\bigg(\frac{l}{l_{\min}}-1\bigg)^2+w_3\left(\frac{\alpha}{\alpha_{\max}}-1\right)^2, \notag\\
\text{subject to,}&& \tilde{\Am}_s(\mathbf{p})\tilde{\Tv}_s=\tilde{\bv}(\mathbf{p}) \label{eq:optprob} \\
&& l_{\min}\leq l\leq l_{\max},\quad \alpha_{\min}\leq \alpha\leq \alpha_{\max},
\end{eqnarray}
where, $\mathbf{p}=(l,\alpha)$ are the design parameters and we have used the fact that $\int_{0}^PT(l,t)dt\approx \sum_{k=1}^{M} T_{Nk}\Delta t$ and $\sum_{k=1}^M T_{Nk}=\langle  \phiv_d,\tilde{\Tv}_s\rangle$ for an appropriate choice of vector $\phiv_d$ which picks the locations of $T_{Nk},k=1, \dots,M$ from the vector $\tilde{\Tv}_s$, and $\phiv_n$ is such that $T_0(l)=\langle \phiv_n,\tilde{\Tv}_s\rangle$. The subscript $s$ in $\tilde{\Tv}_s$ and $\tilde{\Am}_s$  can either be  $f$ or $b$ depending on the discretization scheme used. Finally, note that even though heat equation is a linear PDE, the constraints (\ref{eq:optprob}) are nonlinear in general. Thus, overall the optimization problem involves a quadratic cost function with nonlinear equality and linear inequality constraints, and thus a non-convex problem. 

In~\Cref{sec:bilevel} we propose BVQPCO, a bi-level variational quantum framework to solve the above optimization problem. The VQLS is used to solve the linear system constraint (\ref{eq:optprob}) for given parameters $l,\alpha$, and evaluate the ratio $\frac{\langle \phiv_d,\tilde{\Tv}_s\rangle}{\langle \phiv_n,\tilde{\Tv}_s\rangle}$ required for computing the cost function $C_d$, while a classical black box optimizer, e.g., Bayesian Optimization (BO) is used to select next set of parameters $l,\alpha$ values based on the evaluated cost. Since, VQLS itself involves optimization of the ansatz parameters which parameterizes the solution of the linear system, this leads to a sequence of nested bi-level optimization steps within a variational setting.

\section{Methods}
In this section we describe the general BVQPCO framework and provide an overview of VQLS and BO which are its two main building blocks. We also discuss the specific choice of parameters and methods within VQLS/BO for the heat transfer optimization problem discussed in the~\Cref{sec:heatqen}.

\subsection{BVQPCO Framework}\label{sec:bilevel}
Following the structure of the PDE constrained optimization problem discussed in the~\Cref{sec:heatqen}, we now consider more general optimization problems of the form
\begin{eqnarray}
&& \min_{\pv,\psiv}C_d(\pv,\psiv)\equiv f\left(\frac{\langle \phiv_d,\psiv\rangle}{\langle \phiv_n,\psiv\rangle},\pv\right), \label{eq:classopt}\\
\text{subject to,}&& \Am(\pv)\psiv=\bv(\pv), \label{eq:pdecns} \\
&& g_i(\pv)\leq 0,i=1,\cdots,N_e, \label{eq:parcns}
\end{eqnarray}
where, $\pv\in R^{N_p}$ are the design/optimization variables, $\psiv \in\Rr^{n_s}$ is the discretized solution vector of the underlying PDE represented in the form of the linear constraints (\ref{eq:pdecns}), $g_i(\pv)$ are the constraints on $\pv$, and $C_d$ is the cost function which can expressed as a function $f$ involving $\phiv_d,\phiv_n\in \Rr^{n_s}$ as fixed vectors associated with the cost function.

\begin{figure}[!htbp]
\begin{center}
\includegraphics[scale=0.9]{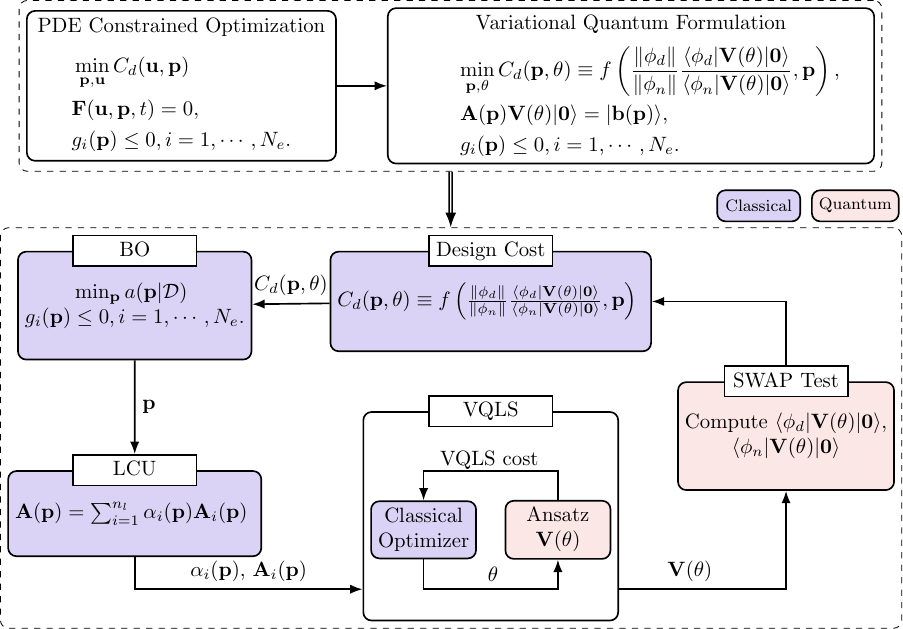}
\caption{Top: Schematic showing transformation of a PDE constrained optimization problem into a variational quantum form. Bottom: Flow diagram of the BVQPCO framework. VQLS uses an inner level optimization to solve the linear system constraints, arising from the discretization of the underlying PDEs, for given design parameters, and evaluate the quantities related to design cost/objective function. A black box optimizer, e.g., BO is used for the outer level optimization to select next set of parameters values based on the evaluated design cost.}
\label{fig:BVQPCO}
\end{center}
\end{figure}

\begin{algorithm}[hbt!]
\begin{algorithmic}[1]
\STATEx Input: $\Am(\pv)$, $\bv(\pv)$, $|\phiv_d\ra$,$|\phiv_n\ra$, $\{g_i(\pv)\}_{i=1}^{N_e}$,$\Vm(\thetav)$,$\gamma$,$n_{sh}$ and $\epsilon$
\STATEx Output: Optimal parameters: $\pv^*$
\STATE Initialize $k=0$ and $\pv^0=\pv_{in}$.
\STATE Determine the unitaries $\Um_d,\Um_n$ which prepare $|\phiv_d\ra$ and $|\phiv_n\ra$, respectively.
\WHILE{stopping criteria not met}
\STATE Compute LCU $\{\alpha_i(\pv^k),\Am_i(\pv^k)\}_{i=1}^{n_k}$ of $\Am(\pv^k)$, and an unitary $\Um_b^{k}$ such that $|\bv(\pv^k)\ra=\Um_b^k|\Zev\ra$.\label{algoline:LCU}
\STATE Compute $\thetav_*^k=VQLS(\{\alpha_i(\pv^k),\Am_i(\pv^k)\}_{i=1}^{n_k},\Um_b^k,\Vm(\thetav),\gamma,n_{sh})$ using the~\Cref{algo:VQLS}. \label{algoline:VQLS}
\STATE Let $|\psiv(\thetav_*^k)\ra=\Vm(\thetav_*^k)|\Zev\ra$. Compute $\la\phiv_d|\psiv(\thetav_*^k)\ra$ and $\la\phiv_n|\psiv(\thetav_*^k)\ra$ using associated quantum circuits, and evaluate the cost $C_d^k=f\left(\frac{\langle \phiv_d|\psiv(\thetav_*^k)\rangle}{\langle \phiv_n|\psiv(\thetav_*^k)\rangle},\pv^k\right)$ . \label{algoline:swap}
\STATE Use classical black box optimizer, e.g. BO~\Cref{algo:BO},  to select next $\pv^{k+1}$ subject to constraints $g_i(\pv^{k+1})\leq 0,i=0,\cdots,N_e$.  \label{algoline:opt}
\STATE Check the stopping criterion, e.g., whether $S(\pv^{k},\pv^{k+1})\leq \epsilon$. \label{algoline:conv}
\STATE $k\leftarrow k+1$
\ENDWHILE
\STATE Return $\pv^k$
\end{algorithmic}
\caption{BVQPCO Algorithm.} \label{algo:bilevel}
\end{algorithm}

To solve this problem in a variational quantum  framework, we propose BVQPCO framework, where the outer optimization level iteratively selects $\pv^k$ using a classical black box optimizer based on the cost function evaluated using solution $|\psiv\ra$ obtained via VQLS for the linear system $\Am(\pv^k)|\psiv\ra=|\bv(\pv^k)\ra$. Thus, the optimization problem (\ref{eq:classopt}) can be expressed in a variational quantum form as
\begin{eqnarray}
&& \min_{\pv,\thetav}C_d(\pv,\thetav)\equiv f\left(\frac{\|\phiv_d\|}{\|\phiv_n\|}\frac{\langle \phiv_d|\Vm(\thetav)|\Zev\ra}{\langle \phiv_n|\Vm(\thetav)|\Zev\ra},\pv\right), \label{eq:classqnt}\\
\text{subject to,} && \Am(\pv)\Vm(\thetav)|\Zev\ra=|\bv(\pv)\ra, \\
&& g_i(\pv)\leq 0,i=1,\cdots,N_e,
\end{eqnarray}
where, $\Vm(\thetav)$ is the VQLS ansatz (see~\Cref{sec:VQLS}) which parametrizes the solution of the linear system (\ref{eq:pdecns}), and $|\phiv_d\ra, |\phiv_n\ra \in R^{n_s}$ define quantum analogues of the cost function parameters $\phiv_d$ and $\phiv_n$, respectively. Note that since,
  \begin{equation}\label{eq:ratio}
    \frac{\langle \phiv_d,\psiv\rangle}{\langle \phiv_n,\psiv\rangle}=\frac{\|\phiv_d\|}{\|\phiv_n\|}\frac{\la \frac{\phiv_d}{\|\phiv_d\|},\frac{\psiv}{\|\psiv\|}\ra}{\langle \frac{\phiv_n}{|\|\phiv_n \|}|\frac{\psiv}{\|\psiv\|}\ra}=\frac{\|\phiv_d\|}{\|\phiv_n\|}\frac{\langle \phiv_d|\psiv\ra}{\langle \phiv_n|\psiv\ra},
\end{equation}
the cost function (\ref{eq:classopt}) can be equivalently expressed in the quantum form (\ref{eq:classqnt}).~\Cref{fig:BVQPCO} shows the overall flow diagram of the BVQPCO framework. The steps to solve (\ref{eq:classqnt})  are described in~\Cref{algo:bilevel}. Some remarks follow:
\begin{itemize}
   \item VQLS in the line (\ref{algoline:VQLS}) is implemented as described in the~\Cref{algo:VQLS}. The inputs include the linear combination of unitaries (LCU) decomposition of $\Am(\pv)$ (see ~\Cref{sec:lcu}), unitary $\Um_b^k$ which prepares $|\bv(\pv^k)\ra$ (see~\Cref{sec:stateprep}),  $\Vm(\thetav)$ is the selected ansatz as discussed above (see~\Cref{sec:ansatz}), $\gamma$ is the stopping threshold (see ~\Cref{sec:costfn}) and $n_{sh}$ is the number of shots used in VQLS cost function evaluation. 
  \item Depending on how $\Am$ depends on the parameters $\pv$, a parameter dependent LCU decomposition $\{\alpha_i(\pv),\Am_i\}$ can be pre computed once, thus saving computational effort. For instance, for the optimization problem (\ref{eq:optprob}), $\Am$ for both the explicit (see~\Cref{eq:Aexp}) and implicit (see~\Cref{eq:Aim}) discretization schemes has a separable structure. Thus, LCU for the matrices $\tilde{\Am}_{f1},\tilde{\Am}_{f2}$ (and similarly $\tilde{\Am}_{b1},\tilde{\Am}_{b2}$) can be computed once, and the LCU of $\Am$ can be readily constructed for any choice of parameters $\pv=(l,\alpha)$.  
   \item For computing, $\la\phiv_d|\psiv(\thetav_*^k)\ra$ and $\la\phiv_n|\psiv(\thetav_*^k)\ra$ in line \ref{algoline:swap}, one can use the circuit associated with SWAP test.
  \item Note that VQLS provides solution of the given linear system only upto a normalization constant. However, due to the assumed structure of the cost function, it is invariant to the normalization constant as shown by the expression~(\ref{eq:ratio}). In fact, there is no loss of generality in the assumed form, i.e., $\frac{\langle \phiv_d,\psiv\rangle}{\langle \phiv_n,\psiv\rangle}$, as one can always choose $\langle \phiv_n,\psiv\rangle$ to depend only the initial condition which is known a priori as part of the problem specification.
  \item For the outer optimization in line (\ref{algoline:opt}) one can use any black box optimization method \cite{alarie2021two}. We propose to use BO (see~\Cref{sec:BO}) which is a sequential design strategy for global optimization of black-box functions that does not require derivatives, is robust to noisy evaluation of the cost function, and uses exploration/exploitation trade off to find optimal solution with minimum number of function calls.

  \item For convergence, line \ref{algoline:conv} uses a step size tolerance, i.e., $S(\pv^{k},\pv^{k+1})=\|\pv^{k+1}-\pv^{k}\|$. However, other conditions can be employed, such as functional tolerance, i.e. the algorithm is terminated when the change $|C_d^{k+1}-C_d^k|$ is within a specified tolerance $\epsilon$.
   
  \item By using VQLS for solving the linear system associated with discretized PDE one can potentially obtain computational advantage over classical methods discussed in the~\Cref{sec:PDEopt}. In the \Cref{sec:erranalysis} we provide a detailed  computational error analysis for VQLS based solution of linear ODEs. Since, there are no theoretical results available for query complexity of VQLS, we employ empirically known results along with our rigorous VQLS error analysis to assess potential advantage of the BVQPCO framework over classical methods.
  \item Note that in the BVQPCO framework one can replace VQLS  by a quantum linear system algorithm (QLSA) \cite{dervovic2018quantum,harrow2009quantum,costa2022optimal}. This could be beneficial since QLSA have provable exponential advantage over classical linear system solvers. However, since QLSA have high quantum resource needs, incorporating QLSA within our proposed framework may only be feasible in a fault tolerant quantum computing setting.    
\end{itemize}

\subsection{Variational Quantum Linear Solver}\label{sec:VQLS}
VQLS is a variational quantum algorithm for solving linear systems of equations on NISQ quantum devices~\cite{VQLS}. Given a linear system of the form
\begin{equation}\label{eq:linsys}
\Am\psiv=\bv,
\end{equation}
where, $\Am\in R^{N\times N}$ and $\psiv,\bv\in R^N$ , the VQLS finds a normalized $|\psi\ra$ to fulfill the relationship
\begin{equation}\label{eq:linsysq}
\Am|\psiv\ra=|\bv\ra,
\end{equation}
where, $|\bv\ra$ denotes the quantum state prepared from the vector $\bv$, and same for $|\psiv\ra$. The inputs to this algorithm are the matrix $\Am$ given as a LCU matrices $\Am_i$ with the coefficients $\alpha_i$
\begin{equation}\label{eq:LCU}
\Am=\sum_{i=1}^{n_l}\alpha_i \Am_i,
\end{equation}
and, a short-depth quantum circuit $\Um_b$ for state preparation $|\bv\ra=\Um_b|\Zev\ra$. Then a cost function $C(\thetav)$ is constructed and evaluated with a devised parameterized ansatz $\Vm(\thetav)$. Through the hybrid quantum-classical optimization loop, the optimal parameters $\thetav$ for the ansatz circuit $\Vm(\thetav)$ can be found such that the cost function $C(\thetav)$ achieves the selected convergence criterion. At the end of this feedback loop, the ansatz $\Vm(\thetav_*)$ prepares the quantum state $|\psiv_*\ra=\Vm(\thetav_*)|\Zev\ra$ which is proportional to the solution $\psiv$ of the linear system (\ref{eq:linsys}).

While no theoretical analysis is available, VQLS is empirically found to scale $O((\log N)^{8.5}\kappa\log(1/\epsilon))$ (i.e., polylogarithmically in matrix size $N$), where $\kappa$ is the condition number of $\Am$, and $\epsilon$ is the desired error tolerance in the solution. Below we summarize comparison of VQLS with classical linear solvers~\cite{patil2022variational}:
\begin{itemize}
 \item Iterative classical methods like conjugate gradient (CG) scale as $O(N s \kappa\log(1/\epsilon))$. VQLS is expected to show advantage over CG for large system sizes, $N >2\times 10^{12} $ ($> 41$ qubits) with sparsity $s = 0.5$, precision $\epsilon= 0.01$ and condition number $\kappa = 1$  However, unlike VQLS the CG method is only applicable for positive definite matrices.
 \item For classical iterative solvers like GMRES (Generalized Minimum Residual Method) which is more general and scales as $O(N^2)$, VQLS would perform better for matrix sizes larger than $ N> 2\times 10^5$ (corresponding to $> 17$ qubits).
\item When we compare VQLS to a direct solver like Gaussian Elimination with $O(N^3)$ complexity we can expect VQLS to show advantage for $N > 750$ ($>9$ qubits).
\end{itemize}

Mainly, VQLS provides an advantage for applications which require only a small part of the full solution vector or only the expectation value of a linear operator acting on the solution vector. Finally, note that like other variational methods, VQLS is plagued by the barren plateau issue which can impose significant computational challenge. 

\subsubsection{LCU Decomposition}\label{sec:lcu}
In general, the matrix $\Am$ is a non-unitary operator. In order to solve the linear system using VQLS, the matrix is decomposed in terms of unitary operators. One popular approach is based on Pauli basis formed from the identity $\In$ and the Pauli gates $\Xm$, $\Ym$ and $\Zm$. A matrix $\Am$ with the size $2^n \times 2^n$ can be written as a linear combination of elements selected from the basis set
\begin{equation*}
P_n=\{\Pm_1\otimes \Pm_2 \dots \otimes \Pm_n:\Pm_i\in\{\In,\Xm,\Ym,\Zm\}\}.
\end{equation*}
with complex coefficients $\alpha_i$ as given in~\Cref{eq:LCU}. In this basis, each $\Am_i\in P_n$ and its corresponding coefficient $\alpha_i$ can be determined via different approaches~\cite{lcu}. For instance, it can be shown that $\alpha_i=Tr(\Am_i\cdot \Am)/2^n$. This matrix multiplication based approach can be computationally very expensive. A more efficient approach based on matrix slicing was proposed in~\cite{lcu} and we adapt it in this work.  For the matrix in \Cref{eq:Aim}, the number of terms scale almost linearly with the matrix size. The decomposition has 14 terms for a $8\times 8$ matrix, 26 terms for $16 \times 16$ matrix and 54 terms for $32 \times 32$ matrix. 

\subsubsection{State Preparation}\label{sec:stateprep}
A normalized complex vector of a quantum state $|\bv\ra$ is prepared by applying a unitary operation
$\Um_b$ to the ground state $|\Zev \ra$, i.e.
\begin{equation*}
|\bv\ra=\Um_b|\Zev\ra.
\end{equation*}
General purpose methods exist to compute the unitary operator $\Um_b$~\cite{shende2005synthesis,mottonen2004transformation}. However, such methods may not always be efficient and can lead to unnecessarily deep circuits which in turn can affect the performance. Instead, we hand-craft circuits to encode states of form shown in ~\Cref{eq:linsysexp} using only controlled rotations.  ~\Cref{fig:stateprep}(a),(b) show the quantum circuits for the state on $3$ and $4$ qubits. The rotation angles are computed using user-defined values of the parameters in the state $\tilde{\bv}$. In contrast, the quantum circuit for the same state constructed using Mottonen state preparation method is deeper as shown in Fig.~ \ref{fig:stateprep} (c).

\begin{figure}
\centering
    \begin{subfigure}[t]{0.5\textwidth}
        \centering
        \includegraphics[scale=0.4]{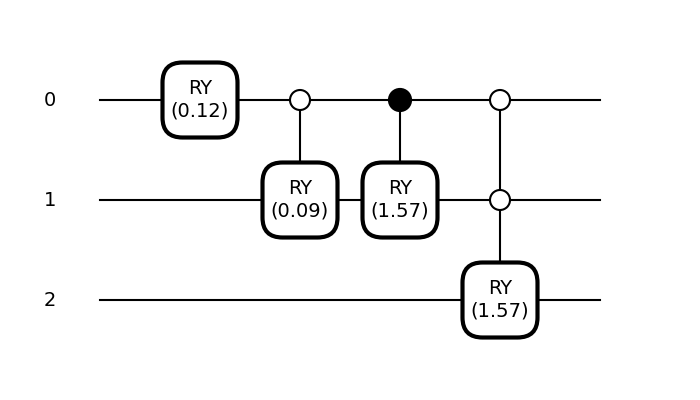}
        \caption{3 qubits, our state preparation method.}
    \end{subfigure}%
    ~ 
    \begin{subfigure}[t]{0.5\textwidth}
        \centering
        \includegraphics[scale=0.4]{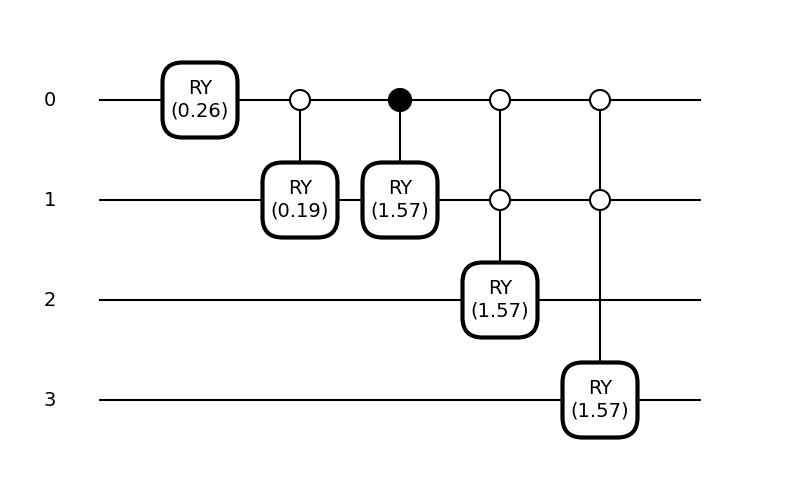}
        \caption{4 qubits, our state preparation method.}
    \end{subfigure}
    
    \begin{subfigure}[t]{1.0\textwidth}
        \centering
        \includegraphics[scale=0.21]{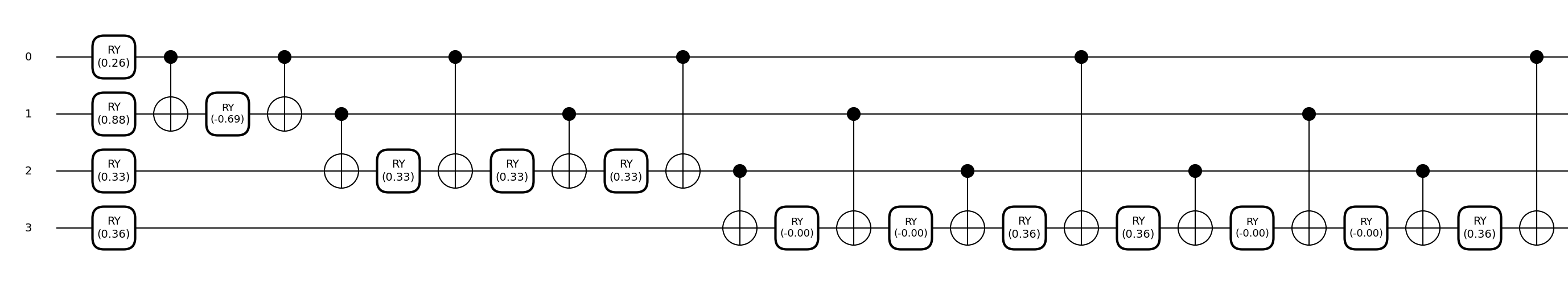}
        \caption{4 qubits, Mottonen state preparation}
    \end{subfigure}
    \caption{Quantum circuit to prepare state $\Tilde{\bv}$ in~\Cref{eq:linsysexp} with  $T_0(x)= 2\big(\frac{\Delta t}{k\Delta x}q^1 \big), q^i = q^1 = 50 \forall i, k=1.0, \Delta t=0.25, \Delta x=0.2$ (a) 3 qubits, $n_x = 2, n_t=4$, and (b) 4 qubits  $n_x = 4, n_t=4$. (c) Mottonen state preparation using built-in PennyLane implementation for 4 qubits, $n_x = 4, n_t=4$ case.}
    \label{fig:stateprep}
\end{figure}

\subsubsection{Ansatz Selection}\label{sec:ansatz}
The VQLS employs a parameterized quantum circuit (PQC) or an ansatz which is a gate sequence $\Vm(\thetav)$ which simulates a potential solution $|\psiv(\thetav)\ra=\Vm(\thetav)|\Zev\ra$. There are variety of choices for ansatz which broadly fall into the categories of fixed-structure ansatz, problem-specific ansatz, and dynamic/adaptive ansatz~\cite{VQLS,patil2022variational}. One popular choice of fixed-structure ansatz is the hardware-efficient ansatz (HEA). HEA are designed without taking into account the specific problem being solved, but rather only the topology (backend connectivity of the qubits) and available gates of a specific quantum computer. While these ansatz can be constructed to be more resistant to noise on any specific available quantum device,, they may not be sufficiently expressive and suffer from barren plateau resulting in trainability issue for large scale problems. Problem-specific ansatz does not take into account the specific quantum device being used, and rather tries to exploit the knowledge of the problem available. The Quantum Alternating Operator Ansatz (QAOA)~\cite{bp_an1} is one such proposed problem-specific ansatz, using two Hamiltonians, known as the driver and the mixer, constructed from specific knowledge of the problem. The requirement of Hamiltonian simulation from the QAOA makes these anatz far less near-term. In variable structure or dynamic ansatz one optimizes over both the gate parameters $\thetav$ and the gate placement in the circuit~\cite{patil2022variational}. 

For real valued problems, such as the example heat transfer optimization problem in~\Cref{sec:heatqen}, the quantum gates used in the ansatz can be restricted appropriately. In this work, we use the fixed structure real-valued ansatz shown in \Cref{fig:ansatz}. This is the modified circuit 9 from~\cite{sim2019expressibility} that was the least expressible with a high KL divergence value of 0.68. 

\begin{figure}[htbp]
\centering
\includegraphics[scale=0.4]{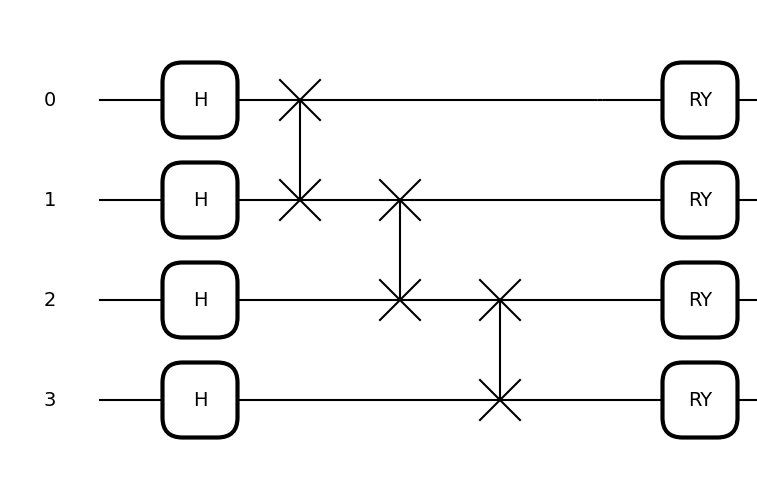}
\caption{Real-valued ansatz (modified circuit 9) on 4 qubits comprising of Hadamard gate (H), SWAP gates and rotation about y-axis gates (RY).}
\label{fig:ansatz}
\end{figure}

\subsubsection{Cost Function}\label{sec:costfn}
The VQLS algorithm aims to minimize the cost function $C_{ug}(\thetav)$  (or its variant as discussed below)
\begin{equation}\label{eq:costun}
C_{ug}(\thetav)=Tr(|\phiv\ra\la\phiv|(\In-|\bv\ra\la\bv|))=\la \psiv|\mathbf{H}_{g}|\psiv\ra,
\end{equation}
where, $|\phiv\ra=\Am|\psiv(\thetav)\ra $ and
\begin{equation}\label{eq:Hg}
\mathbf{H}_{g}=\Am^*\Um_b(\In-|\Zev\ra\la\Zev|)(\Um_b)^*\Am.
\end{equation}
$C_{ug}(\thetav)$ is large when $|\phiv\ra$ is close to being orthogonal to $|\bv\ra$.  Conversely, when $|\phiv\ra$ is nearly proportional to $|\bv\ra$ or small in magnitude,  $C_{ug}(\thetav)$ is small. The latter case does not represent a true solution. To deal with this, the cost function can be normalized as
\begin{equation}\label{eq:costn}
C_{g}(\thetav)=\frac{C_{ug}}{\la\phiv|\phiv\ra}=1-\frac{|\la\bv|\phiv\ra|^2}{\la\phiv|\phiv\ra}.
\end{equation}
Thus, $C_{g}$ can be calculated by computing $\la\bv|\phiv\ra$ and $\la\phiv|\phiv\ra$. Given the LCU for the matrix $\Am$ and the ansatz  $V(\thetav)$, this can be accomplished by noting that
\begin{eqnarray}
  \la\phiv|\phiv\ra &=& \sum_{i=1}^{n_l}\sum_{j=1}^{n_l} \beta_{ij}\alpha_i\alpha_j^*,\quad \beta_{ij}=\la \Zev|\Vm^*\Am_j^*\Am_i\Vm|\Zev\ra, \\
  |\la\bv|\phiv\ra|^2  &=& \sum_{i=1}^{n_l}\sum_{j=1}^{n_l}\gamma_{ij}\alpha_i\alpha_j^*,\quad \gamma_{ij}=\la \Zev|(\Um^b)^*\Am_i\Vm|\Zev\ra \la \Zev|\Vm^*\Am_j^*\Um^b|\Zev\ra.
\end{eqnarray}
There are $n_l(n_l-1)/2$ different $\beta_{ij}$ terms which can be estimated using the Hadamard test.

For calculating the terms $\gamma_{ij}$,  when $i\neq j$, one can compute $\gamma_{ij}$ by applying Hadamard test to estimate $\la\Zev|(\Um_b)^*\Am_i\Vm|\Zev\ra$ and $\la \Zev|\Vm^*\Am_j^*\Um_b|\Zev\ra$. However, this requires controlling all the unitaries, i.e.,  $((\Um_b)^*,\Am_i,\Vm)$. This might be experimentally challenging, especially when the ansatz $\Vm$ consists of many layers. Instead, Hadamard-Overlap test can be used that directly computes $\gamma_{ij}$ without having to control $\Vm$ or $\Um_b$ but at the expense of doubling the number of qubits~\cite{VQLS}. The terms $\gamma_{ii}$ can be  estimated by applying $(\Um_b)^*\Am_i\Vm$ to $|\Zev\ra$ , and then measuring the probability of the all-zeros outcome.

The global cost functions discussed above can exhibit barren plateaus, i.e., cost function gradients vanish exponentially as a function of number of qubits $n$. To improve trainability for large $n$ and mitigate the barren plateau issue, the use of local cost function has been proposed. Such unnormalized and normalized local cost functions can be defined as
\begin{equation}\label{eq:costunloc}
C_{ul}(\thetav)=\la \psiv|\mathbf{H}_{l}|\psiv\ra,\quad C_{l}(\thetav)=\frac{C_{ul}}{\la\phiv|\phiv\ra},
\end{equation}
respectively, where the effective Hamiltonian $\mathbf{H}_l$ is
\begin{equation}
\mathbf{H}_l=\Am^*\Um_b\left(\In-\frac{1}{n}\sum_{i=1}^{n}|0_k\ra\la 0_k|\otimes \In_{\tilde{k}}\right)(\Um_b)^*\Am, \label{eq:Hl}
\end{equation}
with $|0_k\ra$ being a zero state on qubit $k$, and $\In_{\tilde{k}}$ being identity on all qubits except the qubit $k$. Computation of $C_{l}$ can be accomplished via a slight variation of Hadamard test circuit, see~\cite{VQLS} for the details.

Following bounds hold in general for the different cost functions discussed above:
\begin{equation*}
C_{ul}\leq C_{ug}\leq n C_{ul}, \quad C_{l}\leq C_{g}\leq n C_{l},
\end{equation*}
where, $N=2^n$. Consequently, all cost functions vanish under precisely the same condition, namely $|\phiv\ra\propto |\bv\ra$ and thus consistent with each other. Furthermore, assuming $\|\Am\|\leq 1$ the cost functions satisfy (see~\Cref{lem:vqlserr})
\begin{equation}\label{eq:errbounds1}
C_{ug}\geq \frac{\epsilon^2}{\kappa^2} \quad,C_{g}\geq \frac{\epsilon^2}{\kappa^2} \quad,   C_{ul} \geq \frac{\epsilon^2}{n\kappa^2}, \quad C_{l}\ge\frac{\epsilon^2}{n\kappa^2},
\end{equation}
where, $\epsilon=\rho(\psi,\psi(\thetav_*))$ is the error in the normalized solution with $|\psi\ra$ being the exact solution of~\Cref{eq:linsysq}, and  $|\psi(\thetav_*)\ra$ being the approximate VQLS solution. These relations can be used in convergence criterion as discussed later.

The global cost function is used for the numerical studies in~\Cref{sec:num}.  Typically, the evaluation of the global cost function is faster than the local cost function as it requires fewer Hadamard/ Hadamard Overlap tests.  We found that global cost function performed satisfactorily for the system sizes studied in the work. 

\subsubsection{Classical Optimizers}\label{sec:opt}
There are many different classical optimizers available, both gradient-free and gradient-based, that can be used in VQLS for optimization of the ansatz parameters $\thetav$\cite{pellow2021comparison}. Gradient-based methods require gradient information during the optimization process, which can either be obtained analytically by a quantum gradient function, or estimated by means of repeated cost function evaluations. Examples include, limited-memory Broyden–Fletcher–Goldfarb–Shanno (L-BFGS), adaptive gradient algorithm (Adagrad), ADAM and AMSGrad. Gradient-free methods operate as black-box optimizers, which require no extra information besides the values of the cost function evaluations.  Examples in this category include, constrained optimization by linear approximation (COBLYA), Nelder-Mead, Modified Powell method, and simultaneous perturbation stochastic approximation (SPSA).  We experimented with Adagrad and COBLYA in our numerical studies, see Section \ref{sec:num} for the details.

\subsubsection{VQLS Algorithm}
The steps for VQLS are summarized in~\Cref{algo:VQLS}. To initialize the VQLS iterations (see line (\ref{algo:vqlsiniti})) the parameter vector $\thetav_{in}$ needs to be prescribed. Randomly choosing these parameters with highly likelihood will cause VQLS to start in the region of barren plateau which is not desirable. To overcome these, several other methods for initialization can be employed, see for example~\cite{kulshrestha2022beinit,grant2019initialization}.  From the relations (\ref{eq:errbounds1}), one can select threshold $\gamma$ such that, for instance, for $C_{g}\leq \gamma$, results in approximation error $\epsilon\leq \kappa \sqrt{\gamma \log N}$.  The number of shots $n_{sh}$ determines the number of times the quantum circuit is run to estimate each of the term (i.e., $\beta_{ij},\gamma_{ij}$, see the~\Cref{sec:costfn}) in the cost evaluation. Finally, any of the classical optimizers discussed in the previous section can be used in the Step \ref{algo:vqlsopt}.

\begin{algorithm}[hbt!]
\begin{algorithmic}[1]
\STATEx Input: LCU $\{\alpha_i,\Am_i\}_{i=1}^{n_l}$, unitary $\Um_b$ to prepare $|\bv\ra$, ansatz $\Vm(\thetav)$, convergence threshold $\gamma$ and number of shots $n_{sh}$.
\STATEx Output: Optimal parameters: $\thetav_*$.
\STATE Initialize $k=0$, $\thetav^0=\thetav_{in}$ and $c=\infty$. \label{algo:vqlsiniti}
\WHILE{$c\geq \gamma$}
\STATE Compute cost $c=C_g(\thetav^{k})$ using associated quantum circuits with $n_{sh}$ number of shots.
\STATE Use classical optimizer to select next $\thetav^{k+1}$.\label{algo:vqlsopt}
\STATE $k\leftarrow k+1$.
\ENDWHILE
\STATE Return $\thetav_{k}$
\end{algorithmic}
\caption{Outline of VQLS algorithm.} \label{algo:VQLS}
\end{algorithm}

\subsection{Bayesian Optimization}\label{sec:BO}
Bayesian optimization (BO) is an approach for global optimization of black box objective functions that take a long time (minutes or hours) to evaluate~\cite{shahriari2015taking,frazier2018tutorial}.  Specifically, we consider the following problem:
\begin{eqnarray}
&& \min_{\po}\fo(\po)\notag \\
\text{subject to,}&& \gv(\po)\leq 0, \label{eq:Boopt}
\end{eqnarray}
where, $\po\in \Rr^{N_p}$ is a parameter vector, $\fo:\Rr^{N_p}\rightarrow \Rr$ is a scalar objective function, and $\gv:\Rr^{N_p}\rightarrow \Rr^{N_e}$ is the set of $N_e$ constraints. The functional form of $\fo$ is unknown but we can query a  black box function at any point $\po$ and evaluate a noisy value of $\fo(\po)$. The functional form of the constraints $\gv$ is assumed to be known.

BO builds a surrogate model for the objective function $\fo$ and quantities the uncertainty in that surrogate using a Bayesian machine learning technique. It then uses an acquisition function (AF) defined based on this surrogate to decide where to sample next. This process is repeated till selected convergence criteria are met.

\subsubsection{Surrogate Model}
The most popular and effective BO relies on a surrogate model in the form of a Gaussian process (GP) due to its flexibility to represent a prior over function. GP is characterized by a prior mean function $\mu_{\bm}(\cdot)$  and a covariance function $\kappa_{\bk}(\cdot,\cdot)$, where $\bm,\bk$ are the hyperparameters. Consider a finite collection of data pairs $\mathcal{D}=(\Po,\yv)$ of the unknown function $y=f(\po)+\epsv$ with noise $\epsv\sim \mathcal{N}(0,\sigma^2_{\epsilon}(\po))$, where $\Po=[\po_1 \mbox{ } \po_2  \mbox{ }  \cdots  \mbox{ } \po_M]$ are the inputs with corresponding noisy outputs  $\yv=(y_1,\cdots,y_M)^{\tra}$ forming the $M$ training samples. The GP model assumes that the observed data are drawn from a multivariate Gaussian distribution $\mathcal{N}$, and for a new data point $\po$, the joint distribution of the observed outputs $\yv$ and the predicted output $y$ are:
\begin{equation*}
\left(
  \begin{array}{c}
    y \\
    \yv \\
  \end{array}
\right)\sim \mathcal{N}\left(\cdot,\left(\begin{array}{c}\mu_{\bm}(\po)\\
                                                  \mu_{\bm}(\Po)\end{array}\right),\left[ \begin{array}{cc}
                                       \kappa_{\bk}(\po,\po)  & K^{\tra}_{\bk}(\Po,\po) \\
                                        K_{\bk}(\Po,\po) & K_{\bk}(\Po,\Po)+\Sigma_{\epsilon}
                                      \end{array}
\right]\right),
\end{equation*}
where,
\begin{itemize}
  \item $\mu_{\bm}(\Po)=[\mu_{\bm}(\po_1) \mbox{ } \mu_{\bm}(\po_2) \cdots \mu_{\bm}(\po_M)]^{\tra}$ in a $M\times 1$ vector,
  \item $K_{\bk}(\Po,\Po)=[\kappa_{\bk}(\po_i,\po_j)]$ is a $M\times M$ correlation matrix evaluated at each pair of the training points,
  \item $K_{\bk}(\Po,\po)=[\kappa_{\bk}(\po_1,\po) \mbox{ } \kappa_{\bk}(\po_2,\po) \cdots \kappa_{\bk}(\po_M,\po)]^{\tra}$ denotes a $M \times 1$ correlation vector evaluated at all pairs of training and test point,
  \item and $\Sigma_{\epsilon}=\mbox{diag}(\sigma^2(\po_1),\cdots,\sigma^2(\po_M))$ is a diagonal matrix.
\end{itemize}
The conditional distribution $p(y|\po,\Po,\yv)\sim \mathcal{N}(\mu(\po),\sigma^2(\po))$ is then a multivariate Gaussian distribution, where the mean and variance of the predicted output $y$ are given by
\begin{align}
  \mu(\po) &=\mu_{\beta_m}(\po)+ K_{\bk}(\po,\Po)(K_{\bk}(\Po,\Po)+\Sigma_{\epsilon})^{-1}(\yv-\mu_{\beta_m}(\Po)), \notag \\
  \sigma^2(\po) &= \kappa_{\bk}(\po,\po)-K_{\bk}^\prime(\Po,\po)(K_{\bk}(\Po,\Po)+\Sigma_{\epsilon})^{-1} K_{\bk}(\Po,\po), \label{eq:boup}
\end{align}
respectively. Commonly, the mean function $\mu_{\bm}(\cdot)$ is taken to be a zero function, and the kernel function is selected as the squared exponential kernel or the Matáern kernel with hyperparameters, such as length scale, signal variance, and noise variance. We will denote by $\bo$ all the hyperparameters involved. Given the data $\mathcal{D}$, the optimal hyperparameters $\bo^*$ are inferred using MLE techniques such as by maximizing the log marginal likelihood. Alternatively, MAP and full Bayesian approaches can also be applied though they tend to be computationally more expensive~\cite{williams2006gaussian}.

\subsubsection{Acquisition Function}
Given the surrogate, AFs are the utility functions $\aq(\po|\mathcal{D})$ that guide the search to reach the optimum of the objective function by identifying where to sample next. The guiding principle behind AFs is to strike a balance between exploration and exploitation, which is achieved by querying samples from both known high-fitness-value regions and regions that have not been sufficiently explored so far. Also AF should be computable readily and easy to optimize. Note that any constraints on $\po$ can be incorporated while optimizing the AF, i.e.
\begin{eqnarray}
&& \min_{\po}\aq(\po|\mathcal{D})\notag\\
\text{subject to,} && \gv(\po)\leq 0. \label{eq:aqopt}
\end{eqnarray}
Thus, each selected point $\po_{M+1}$ via above optimization is feasible.

Some popular choices of AF include probability of improvement, expected improvement (EI), knowledge gradient, and information theoretic measures, see~\cite{frazier2018tutorial}. We will use EI and its noisy extension for our application. Given the data $\mathcal{D}$, for noiseless case, i.e. $\epsilon=0$, EI is defined as
\begin{equation*}
\aq_{EI}(\po|\mathcal{D})=E_{y\sim \mathcal{N}(\cdot, \mu(\po),\sigma^2(\po))}[\max(0,\fo^*-y)],
\end{equation*}
where, $\fo^*=\min\{y_i|y_i\in \mathcal{D}\}$ and $\mu(\po)$ and $\sigma(\po)$ are as defined in~\Cref{eq:boup} with $\Sigma_{\epsilon}=0$. This expectation has a closed form expression
\begin{equation}
\aq_{EI}(\po|\mathcal{D})=\sigma(\po)z\Phi(z)+\sigma(\po)\phi(z), \label{eq:EI}
\end{equation}
where, $z=\frac{\fo^*-\mu(\po)}{\sigma(\po)}$, and $\Phi$ and $\phi$ are the cumulative density function and the density function of the standard Normal distribution, respectively. Thus, EI is easy to implement and optimize.

For noisy case $\epsilon\neq 0$, $f^*$ cannot be evaluated exactly since only noisy values of $\fo_i$ are available. There are various approaches to handle this case. We use the formulation from~\cite{letham2019constrained}, where EI is replaced by expected EI, defined as:
\begin{equation*}
\aq_{NEI}(\po|\mathcal{D})=\int_{\fov}\aq_{EI}(\po|\mathcal{D}_{\fov})p(\fov|\mathcal{D})d\fov,
\end{equation*}
where, $\fov=(f_1,\cdots,f_M)^{\tra}$,  $p(\cdot|\mathcal{D})=\mathcal{N}(\cdot,\mu_{\bm}(\Po),K_{\bk}(\Po,\Po)+\Sigma_{\epsilon})$ is multivariate Gaussian distribution based on noisy data $\mathcal{D}$, and $\aq_{EI}(\po|\mathcal{D}_{\fov})$ is evaluated based on a noiseless GP model fitted to data $\mathcal{D}_{\fov}=\{(\po_i,f_i)\}$.

This function does not have an analytic expression, but can be estimated using Monte Carlo methods. Let $K_{\bk}(\Po,\Po)+\Sigma_{\epsilon}=\Bm\Bm^\tra$, e.g., obtained via Cholskey decomposition. Then
\begin{equation*}
\aq_{NEI}(\po|\mathcal{D})\approx\frac{1}{N_m}\sum_{1}^{N_m}\aq_{EI}(\po|\mathcal{D}_{\fov_i}),
\end{equation*}
where, $\fov_i=\Bm\Phi^{-1}(\sv_i)+\mu_{\bm}(\Po)$ with $\sv_i\in [0,1]^{N_p}$ is a random sample in unit cube of dimension $N_p$. For sampling $\sv_i$ one can use Monte Carlo or Quasi Monte Carlo (QMC) techniques, with latter being preferred as it gives more uniform coverage of the samples. We used a variation of this technique known as the qNoisy EI ~\cite{balandat2020botorch}, where the samples are taken from the joint posterior over $q$ test points and the previously observed points.

\subsubsection{BO Algorithm}
Overall steps of BO algorithm are described in~\Cref{algo:BO}. In Step~\ref{bo:init}, $N_o$ QMC samples of input parameters are chosen and the cost function is evaluated $N_s$ times for each sample, to determine the mean and variance of the cost function.  In Step~\ref{bo:MLE} MLE is applied to determine GP hyperparameters $\beta$, and the GP model is updated in Step~\ref{bo:up}. The AF is optimized in Step~\ref{bo:afopt} to determine the next sampling point. In Step \ref{bo:ev} the black box function is queried $N_s$ times to evaluate mean/variance of the cost function at the new sampling point. Finally, if the convergence criteria are not met in Step~\ref{bo:conv}, the iterations repeat from Step~\ref{bo:MLE}. The convergence criteria can be based on stepsize or functional tolerance, similar to as discussed for the VQLS. Finally, note that if the error in evaluating the cost function is negligible, i.e. $\epsilon\approx 0$, there is no need to compute variance in cost function evaluation. Therefore, $N_s$ can be taken to be a single sample, $\Sigma_{\epsilon}$ is set to be zero in the BO update (\ref{eq:boup}), and the noiseless EI formula (\ref{eq:EI}) can be used for the AF optimization in Step~\ref{bo:afopt}. 

\begin{algorithm}[hbt!]
\begin{algorithmic}[1]
\STATEx Input: GP mean/kernel functions and associated parameters $\beta$, black box function $\fo$, number of initial samples $N_o$, number of noisy function evaluation $N_s$, convergence criteria.
\STATEx Output: Optimal parameters: $\po_*$
\STATE Randomly select a set of QMC points $\po_i,i=1,\cdots,N_o$. Let $y_{ij},j=1\cdots,N_s$ be samples of noisy values of $f(\po_i)$ evaluated at $\po_i$. Let
\begin{eqnarray*}
  \overline{y}_i &=& \frac{1}{N_s}\sum_{j-1}^{N_s}y_{ij}, \quad \Delta_i = \frac{1}{N_s}\sum_{j-1}^{N_s} (y_{ij}-\overline{y}_i)^2,
\end{eqnarray*}
be the estimated mean and variance at $\po_i$, respectively. Let $\mathcal{D}=\{(\po_i,\overline{y}_i,\Delta_i)\}$.\label{bo:init}
\STATE Given $\mathcal{D}$ use MLE to infer the prior hyperparameters $\beta$. \label{bo:MLE}
\STATE Update the GP posterior mean and covariance using~\Cref{eq:boup}, with $\Sigma_{\epsilon}=\mbox{diag}(\Delta_1,\cdots,\Delta_{N_o})$. \label{bo:up}
\STATE Optimize the acquisition function $\aq_{NEI}(\po|\mathcal{D})$ as per (\ref{eq:aqopt}) and determine next sampling point $\po_{n}$. \label{bo:afopt}
\STATE Compute $\overline{y}_n$ and associated $\Delta_n$ at $\po_n$. Augment the data $\mathcal{D}\leftarrow \mathcal{D}\bigcup \{\po_n,\overline{y}_n,\Delta_n\}$. \label{bo:ev}
\STATE If convergence criterion not met, go to the Step \ref{bo:MLE}. \label{bo:conv}
\end{algorithmic}
\caption{Outline of BO Algorithm.} \label{algo:BO}
\end{algorithm}

\section{Computational Error and Complexity Analysis}
\label{sec:erranalysis}

In this section we present a detailed complexity analysis for VQLS based solution of linear ODEs under explicit and implicit Euler discretization schemes, and use that to assess potential advantage of BVQPCO framework over classical techniques.  
Consider the initial value problem (IVP) for an inhomogeneous system of linear ODEs
\begin{eqnarray}
\dot{\uv}&=&\Am \uv+\bv(t), \label{eq:qdc}\\
\uv(0)&=&\uv_0\in\Rr^\nx,\notag
\end{eqnarray}
where, $\Am\in \Rr^{\nx\times\nx}$, $\bv(t)\in\Rr^{\nx}$, $\uv(t)\in \Rr^\nx$ and $t\in[0,T]$. The above ODE for instance could arise from spatial discretization of a linear PDE, as illustrated in~\Cref{sec:heatqen}. We shall denote by $\uv_c$ as the vector of solution 
\begin{equation}
\uv_c=\left(\begin{array}{cc}
     \uv(0)  \\
     \uv(h)   \\
     \vdots \\
     \uv(h(M-1))
\end{array}\right),\label{eq:uc}
\end{equation}
sampled at $t_k=(k-1)h,k=1,\cdots,M$ for a given sampling step size $h>0$ and an integer $M>0$.

\begin{assumption}\label{assumption1}
For system (\ref{eq:qdc}), we assume
\begin{itemize}
  \item A1: $\Am$ is diagonisable with eigenvalues $\lambda_i,i=1,\cdots,\nx$ satisfying
  \begin{equation}\label{eq:lam}
    Re(\lambda_i)< 0.
  \end{equation}
  Furthermore, if $\Am$ has real eigenvalues, the assumption can be relaxed to $\lambda_i \leq 0$.
  \item A2: $\|\bv(t)\|=\sup_{[0,T]}\|\bv(t)\|_2$ and $\|\dot{\bv}(t)\|=\sup_{[0,T]}\|\dot{\bv}(t)\|_2$ are bounded.
\end{itemize}
\end{assumption}
We will denote by $\rho(\Am)$ as the spectral radius of $\Am$, and let $\nu=\|\Am\|_2$ be its spectral norm.

Before providing detailed error analysis for VQLS based solution of this IVP, we present two lemmas which will be necessary for the proofs that follow.

\begin{lemma}\label{lem:vqlserr}
For the linear system (\ref{eq:linsys}), following bounds hold for the different VQLS cost functions
\begin{equation}\label{eq:errbounds}
C_{ug}\geq \frac{\epsilon^2}{\kappa^2} \quad,C_{g}\geq \frac{\epsilon^2}{\kappa^2\|\Am\|} \quad,   C_{ul} \geq \frac{\epsilon^2}{n\kappa^2}, \quad C_{l}\ge\frac{\epsilon^2}{n\kappa^2\|\Am\|},
\end{equation}
where, $\kappa$ and $N=2^n$ are the condition number and size of $\Am$ in (\ref{eq:linsys}) respectively, and $\epsilon$ is the error tolerance 
\begin{equation*}
\epsilon=\rho(\psi,\psi(\thetav^*)),
\end{equation*}
with $\rho$ being the trace norm,  $|\psi\ra$ being the exact solution of~\Cref{eq:linsysq}, and $|\psi(\thetav_*)\ra$ being the approximate VQLS solution.
\end{lemma}
The proof of above result can be found in~\cite{VQLS}.

\begin{lemma}\label{eq:relnorms}
The trace norm and $l_2$ norm between $|\psiv\ra$ and $|\phiv\ra$ are related as follows:
\begin{equation}\label{eq:relnorms1}
\left(1-\frac{\||\psiv\ra-|\phiv\ra\|_2^2}{2}\right)^2+\rho^2(\psiv,\phiv)=1.
\end{equation}
\end{lemma}

\begin{proof}
Without loss of generality, we can write $|\psiv\ra$ in terms of  $|\phiv\ra$ and state  $|\phiv^{\perp}\ra$ orthogonal to it as
\begin{equation*}
|\psiv\ra=\cos\theta|\phiv\ra +\sin\theta |\phiv^{\perp}\ra.
\end{equation*}
Then as shown in~\cite{wilde2011classical} (see pages 274-275)
\begin{equation*}
\rho(\psiv,\phiv)=|\sin \theta|.
\end{equation*}
Also
\begin{equation*}
\||\psiv\ra -|\phiv\ra\|^2=2-2\cos\theta.
\end{equation*}
Finally, using the identity $\cos^2\theta+\sin^2\theta=1$, leads to the desired result (\ref{eq:relnorms1}).
\end{proof}

\begin{assumption}\label{assumption2}
We assume that the query complexity of VQLS scales as $\mathcal{O}(\log^{8.5} N\kappa\log(1/\epsilon))$, where $\kappa$ is the condition number and $N$ is size of the matrix $\Am$, and $\epsilon$ is the desired accuracy of the VQLS solution as described in~\Cref{lem:vqlserr}.
\end{assumption}
Note that above assumption is based on an empirical scaling study in~\cite{VQLS} and no theoretical guarantees are available. 

\subsection{Error Analysis for VQLS Based Explicit Euler Linear ODE Solver}
The application of the forward Euler discretization scheme to the IVP (\ref{eq:qdc}) involves following steps.

\paragraph{Step 1:} Applying the forward Euler method to the system with step size $h$ yields the difference equation
\begin{equation}\label{eq:euler}
\uvh_f^{k+1}=(\Id+\Am h)\uvh_f^k+h\bv^k, \quad k=1,\cdots,\tsteps-1,
\end{equation}
where, $\uvh_f^k$ approximates $\uv(t_k)=\uv((k-1)h)$, with $\uvh_f^1=\uv(0)=\uv_0$, $\bv^k=\bv((k-1)h)$ and  $\tsteps-1=T/h$. The error introduced by Euler discretization is characterized by~\Cref{lem:eulerex}.

\paragraph{Step 2:}  The iterative system (\ref{eq:euler}) can be expressed as a system of linear equations
\begin{equation}\label{eq:linp}
\overbrace{\left(
  \begin{array}{cccc}
    \Id & 0 & 0 & \cdots \\
    -[\Id+\Am h] & \Id & 0 & \cdots \\
    0 & 0 & \ddots & \ddots \\
    0 & 0 &  -[\Id+\Am h] & \Id \\
  \end{array}
\right)}^{\tilde{\Am}_f}\overbrace{\left(\begin{array}{c}
         \uvh_f^1 \\
         \uvh_f^2 \\
         \vdots \\
         \uvh_f^{\tsteps} \\
       \end{array}\right)}^{\tilde{\uv}_f}
=\overbrace{\left(\begin{array}{c}
         \uv_{0} \\
         h\bv^1 \\
         \vdots \\
         h\bv^{\tsteps-1} \\
       \end{array}\right)}^{\tilde{\bv}}.
\end{equation}
In the VQLS framework the linear system (\ref{eq:linp}) is further transformed into the form
\begin{equation}
\tilde{\Am}_f|\tilde{\uv}_f\ra=|\tilde{\bv}\ra,\label{eq:linqp}
\end{equation}
where, $|\tilde{\bv}\ra=\tilde{\bv}/\|\tilde{\bv}\|$ and $|\tilde{\uv}_f\ra=\tilde{\uv}_f/\|\tilde{\uv}_f\|$ are the normalized vectors. The VLQS algorithm then optimizes the parameter $\thetav$ of the ansatz $\Vm(\thetav)$ such that
\begin{equation*}
\tilde{\Am}_f\Vm(\thetav)|\Zev\ra=|\tilde{\bv}\ra.
\end{equation*}
The optimal parameter $\thetav_*$ can be used to prepare a solution $|\uvo_f\ra=\Vm(\thetav_*)|\Zev\ra$ which is an approximation to $|\tilde{\uv}_f\ra$. The approximation error is characterized by the~\Cref{lem:vqlserr}.~\Cref{fig:errfig} shows a schematic of the steps involved in transforming and solving a given PDE using VQLS and the associated variables/approximation errors.

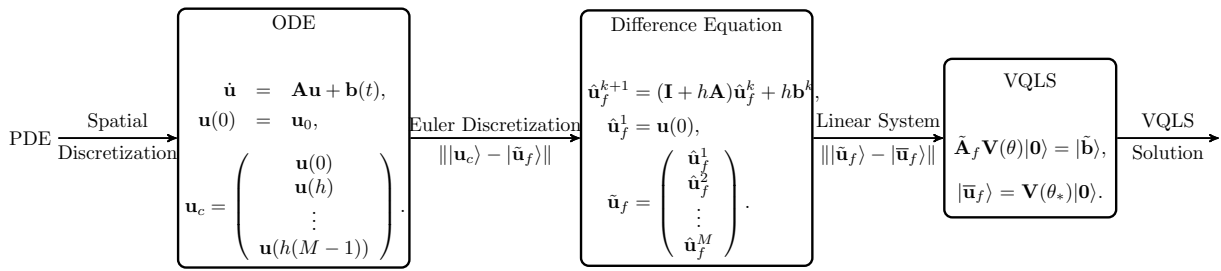
\begin{figure}[!htbp]
\begin{center}
\input{Figure2.tikz}
\caption{Schematic showing different stages in going from the PDE to the VQLS solution, associated variables and the error in approximations. The case of explicit/forward Euler scheme is shown as an example. Similar flow applies for the case of implicit/backward Euler scheme.}\label{fig:errfig}
\end{center}
\end{figure}

\begin{restatable}{lemma}{hchoice}\label{lem:kappa}
Let $h$ be chosen such that $h\nu\leq 2$, where $\nu=\|\Am\|_2$. Then the stiffness $\kappa$ of $\tilde{\Am}_f$ in (\ref{eq:linqp}) is bounded as follows
\begin{equation}
\label{eq:kappa}
\kappa(\tilde{\Am}_f)\leq \frac{12 e^{\nu T}}{h\nu}.
\end{equation}
\end{restatable}

\begin{restatable}{lemma}{eulerex}\label{lem:eulerex}
Let $\epsilon>0$ be the desired error, then one can choose the step size $h$ in the forward Euler scheme (\ref{eq:euler}), such that
\begin{equation}\label{eq:errexp1}
\||\uv_c\ra-|\tilde{\uv}_f\ra\|\leq \epsilon,
\end{equation}
where, $\uv_c$ is as defined in (\ref{eq:uc}) with this selected $h$ and $|\tilde{\uv}_f\ra$ is solution of the system (\ref{eq:linqp}).
\end{restatable}
See~\Cref{sec:proof_lem3,sec:proof_lem4} for proofs of the above two lemmas.

\begin{theorem}\label{thm:errfor}
Consider the IVP (\ref{eq:qdc}) under~\Cref{assumption1}.  Let $0<\epsilon<1$ be the desired solution accuracy, then one can choose the step size $h$ in the forward Euler scheme (\ref{eq:euler}), and $\gamma$ the stopping threshold in VQLS Algorithm~\ref{algo:VQLS} such that
\begin{equation}\label{eq:errexp2}
\| |\uv_c\ra-|\uvo_f\ra\|\leq \epsilon,
\end{equation}
where, $|\uvo_f\ra=\Vm(\thetav^*)|\Zev\ra$ is solution generated by applying VQLS to the linear system (\ref{eq:linqp}). Furthermore, under the Assumption \ref{assumption2} the query complexity scales as 
\begin{equation}
 \mathcal{O}\left(\log^{8.5} \left(\frac{\nx T^2}{\|\uv_c\|^2\epsilon^2}\right)\frac{T e^{\nu T}}{\|\uv_c\|^2\epsilon^2}\log(1/\epsilon)\right).
\end{equation}
\end{theorem}

\begin{proof}
From Lemmas \ref{lem:eulerex} and \ref{lem:kappa}, we choose $h$
\begin{equation*}
h<\min\left\{\frac{\epsilon^2\|\uv_c\|^2}{16C^2T},\frac{2}{\nu}\right\},
\end{equation*}
such that
\begin{equation}\label{eq:err1p}
\||\uv_c\ra-|\tilde{\uv}_f\ra\|\leq \epsilon/2.
\end{equation}
Let $\gamma$ be
\begin{equation}\label{eq:gamma}
\gamma=\frac{1-(1-\frac{\epsilon^2}{8})^2}{\left(\frac{12e^{\nu T}}{h\nu}\right)^2(\log(\nx \tsteps))(2+h\nu)}<1,
\end{equation}
where, $N,M>1$. Let the VQLS algorithm be terminated under the condition (with similar analysis for other VQLS cost functions)
\begin{equation*}
C_{g}\leq \gamma,
\end{equation*}
then from the Lemmas \ref{lem:vqlserr} and \ref{lem:kappa}
\begin{equation}
\rho^2(|\tilde{\uv}_f\ra,|\uvo_f\ra)=(\epsilon^\prime)^2\leq n\kappa^2\|\tilde{\Am}_f\|\gamma \leq 1-\left(1-\frac{\epsilon^2}{8}\right)^2\leq \frac{\epsilon^2}{8}(2-\frac{\epsilon^2}{8})\leq \frac{\epsilon^2}{4},\label{eq:epserr}
\end{equation}
where, the number of qubits $n=\log(\nx\tsteps)$ and we have used the inequality (\ref{eq:normAF}) .
It then follows from~\Cref{eq:relnorms}
\begin{equation*}
1-\left(1-\frac{\||\tilde{\uv}_f\ra-|\uvo_f\ra\|_2^2}{2}\right)^2\leq 1-\left(1-\frac{\epsilon^2}{8}\right)^2,
\end{equation*}
which implies
\begin{equation}\label{eq:err2p}
\||\tilde{\uv}_f\ra-|\uvo_f\ra\|_2\leq \epsilon/2.
\end{equation}
Thus, using the relations (\ref{eq:err1p}) and (\ref{eq:err2p})
\begin{equation}
\| |\uv_c\ra-|\uvo_f\ra\|\leq \| |\uv_c\ra-|\tilde{\uv}_f\ra\|+\||\tilde{\uv}_f\ra-|\uvo_f\ra\|\leq \epsilon.
\end{equation}

Furthermore, under~\Cref{assumption2} the query complexity of VLQS scales as $\mathcal{O}(\log^{8.5} (\nx\tsteps)\kappa\log(1/\epsilon^\prime))$ which can be simplified as follows
\begin{eqnarray}
 && \mathcal{O}(\log^{8.5} (\nx\tsteps)\kappa\log(1/\epsilon^\prime)) = \mathcal{O}\left(\log^{8.5}(\nx T/ h)\frac{e^{\nu T}}{h\nu}\log(1/\epsilon)\right)\notag\\
  &=&  \mathcal{O}\left(\log^{8.5}\left(\nx T\left(\frac{\nu}{2}+\frac{4C^2T}{\|\uv_c\|^2\epsilon^2}\right)\right)e^{\nu T}\left(\frac{1}{2}+\frac{4C^2T}{\|\uv_c\|^2\epsilon^2}\right)\log(1/\epsilon)\right)\notag\\
  &=& \mathcal{O}\left(\log^{8.5} \left(\frac{\nx T^2}{\|\uv_c\|^2\epsilon^2}\right)\frac{T e^{\nu T}}{\|\uv_c\|^2\epsilon^2}\log(1/\epsilon)\right) \label{eq:complfe}.
\end{eqnarray}

\end{proof}
\begin{remark}
From the estimate of query complexity, one can see that while it scales polylogarithmically with the system size $\nx$, there is an exponential dependence on time period of integration, i.e. $T$. The next theorem shows that under the additional assumption that $\Am$ is unitarily diagonalizable (i.e., $\Am$ is a normal matrix), the exponential dependence on $T$ reduces to polynomial dependence.    
\end{remark}

\begin{restatable}{theorem}{normat}\label{thm:normat}
Consider the IVP (\ref{eq:qdc}) under~\Cref{assumption1}. Additionally, assume that $\Am$ in  (\ref{eq:qdc}) is a normal matrix. Let $0<\epsilon<1$ be the desired solution accuracy, then one can choose the step size $h$ in the forward Euler scheme (\ref{eq:euler}), and $\gamma$ the stopping threshold in VQLS algorithm (\ref{algo:VQLS}) such that
\begin{equation}
\| |\uv_c\ra-|\uvo_f\ra\|\leq \epsilon,
\end{equation}
where, $|\uvo_f\ra=\Vm(\thetav^*)|\Zev\ra$ is solution generated by applying VQLS to the linear system (\ref{eq:linqp}). Furthermore, under the Assumption \ref{assumption2} the query complexity scales as
\begin{equation}
\mathcal{O}\left(\log^{8.5} \left(\frac{\nx T^4}{\|\uv_c\|^2\epsilon^2}\right)\frac{T^3}{\|\uv_c\|^2\epsilon^2}\log(1/\epsilon)\right)\label{eq:complfenor}.
\end{equation}
\end{restatable}
The proof is given in~\Cref{sec:proof_thm2}.

\begin{remark}
The matrix $\Am$ (see~\Cref{eq:A}) that arises in the discretization of heat equation is symmetric and hence a normal matrix. 
\end{remark} 

\subsubsection{Comparison with Classical Explicit Euler Linear ODE Solver}\label{sec:cmpexp}
Given the update eqn. (\ref{eq:euler}),  the computational complexity of classical explicit Euler method can be estimated as
\begin{equation*}
O(sNM)=O(sNT/h)
\end{equation*}
where, $s$ is sparsity of $\Am$. Following the analysis in Lemma \ref{lem:eulerex}, we choose $h$ such that
\begin{equation*}
h<\min\left\{\frac{\epsilon^2\|\uv_c\|^2}{C^2T},\frac{2}{\nu}\right\}.
\end{equation*}
Under this condition, the classical Euler approach generates a normalized solution  vector $|\tilde{\uv}^c_f\ra$, such that
\begin{equation}\label{eq:err1class}
\| |\uv_c\ra-|\tilde{\uv}^c_f\ra\|\leq \epsilon,
\end{equation}
consistent with the VQLS approach. Thus computational complexity of classical explicit Euler scheme scales as
\begin{equation*}
\mathcal{O}\left(sN\frac{T}{h}\right)=\mathcal{O}\left(sNT\left(\frac{\nu}{2}+\frac{C^2T}{\|\uv_c\|^2\epsilon^2}\right)\right)=\mathcal{O}\left(\frac{sNT^2}{\|\uv_c\|^2\epsilon^2}\right).
\end{equation*}
On the other hand, based on~\Cref{thm:errfor,thm:normat}, the query complexity of VQLS based explicit Euler scheme scales polylogarithmically w.r.t $N$.   

We next discuss implications of this result from the perspective of BVQPCO framework. We assume that, as in  BVQPCO framework, an outer optimization loop for design optimization is used with the underlying classical explicit Euler method. Further assuming that number of outer loop iterations are similar in both this classical and the BVQPCO framework, the query complexity of BVQPCO will scale polylogarithmically with $N$  under~\Cref{assumption2}, and thus could provide a significant computational advantage for simulation based design problems.

\subsection{Error Analysis for VQLS Based Implicit Euler Linear ODE Solver}
The application of the implicit Euler discretization scheme to the IVP (\ref{eq:qdc}) involves following steps.

\paragraph{Step 1:} Applying the backward Euler method to the system with step size $h$ yields,
\begin{equation}\label{eq:eulerb}
(\Id-\Am h)\uvh_b^{k+1}=\uvh_b^k+h\bv^k, \quad k=1,\dots,\tsteps-1,
\end{equation}
where, $\uvh_b^k$ approximates $\uv(t_k)=\uv((k-1)h)$, with $\uvh_b^1=\uv(0)=\uv_0$, $\bv^k=\bv((k-1)h)$ and  $\tsteps-1=T/h$. The error introduced by Euler discretization is characterized by~\Cref{lem:eulerexb}.

\paragraph{Step 2:}  The iterative system (\ref{eq:eulerb}) can be expressed as a system of linear equations,
\begin{equation}\label{eq:linpb}
\overbrace{\left(
  \begin{array}{cccc}
    \Id & 0 & 0 & \cdots \\
    -\Id & [\Id-\Am h] & 0 & \cdots \\
    0 & 0 & \ddots & \ddots \\
    0 & 0 &  -\Id & [\Id-\Am h] \\
  \end{array}
\right)}^{\tilde{\Am}_b}\overbrace{\left(\begin{array}{c}
         \uvh_b^1 \\
         \uvh_b^2 \\
         \vdots \\
         \uvh_b^{\tsteps} \\
       \end{array}\right)}^{\tilde{\uv}_b}
=
\overbrace{\left(\begin{array}{c}
         \uv_{0} \\
         h\bv^1 \\
         \vdots \\
         h\bv^{\tsteps-1} \\
       \end{array}\right)}^{\tilde{\bv}},
\end{equation}
which in the VQLS framework is further transformed into the form,
\begin{equation}
\tilde{\Am}_b|\tilde{\uv}_b\ra=|\tilde{\bv}\ra,\label{eq:linqpb}
\end{equation}
where, $|\tilde{\uv}_b\ra=\tilde{\uv}_b/\|\tilde{\uv}_b\|$ is the normalized vector. We shall denote by $|\uvo_b\ra=\Vm(\thetav^*)|\Zev\ra$, the solution produced by VQLS.

\begin{restatable}{lemma}{kappab}\label{lem:kappab}
Let $h$ be chosen such that $h\nu\leq 0.5$, where $\nu=\|\Am\|$. Then the stiffness $\kappa$ of $\tilde{\Am}_b$ in (~\ref{eq:linqpb}) is bounded as follows
\begin{equation}\label{eq:kappab1}
\kappa(\tilde{\Am}_b)\leq 2.5(T/h+1) ,
\end{equation}
if $\|(\Id-h\Am)^{-1}\|\leq 1$ or 
\begin{equation}\label{eq:kappab2}
\kappa(\tilde{\Am}_b)\leq  \frac{5e^{2\nu T}}{h\nu},
\end{equation}
if  $\|(\Id-h\Am)^{-1}\|>1$.
\end{restatable}

\begin{restatable}{lemma}{eulerexb}\label{lem:eulerexb}
Let $\epsilon>0$ be the desired error, the one can choose the step size $h$ in implicit Euler scheme (\ref{eq:eulerb}), such that
\begin{equation}\label{eq:errimp1}
\||\uv_c\ra-|\tilde{\uv}_b\ra\|\leq \epsilon,
\end{equation}
where, $\uv_c$ is as defined in (\ref{eq:uc})  with this selected $h$ and $|\tilde{\uv}_b\ra$ is solution of the system (\ref{eq:linqpb}).
\end{restatable}
For the proofs of above two lemmas see~\Cref{sec:proof_lem5,sec:proof_lem6}.

\begin{theorem}\label{thm:errback}
Consider the IVP (\ref{eq:qdc}) under the Assumptions \ref{assumption1}.  Let $0<\epsilon<1$ be the desired solution accuracy, then one can choose the step size $h$ in the forward Euler scheme (\ref{eq:eulerb}), and $\gamma$ the stopping threshold in VQLS algorithm (\ref{algo:VQLS}) such that
\begin{equation}\label{eq:errimp2}
\| |\uv_c\ra-|\uvo_b\ra\|\leq \epsilon,
\end{equation}
where, $|\uvo_b\ra=\Vm(\thetav_*)|\Zev\ra$ is solution generated by applying VQLS to the linear system (\ref{eq:linqpb}). Furthermore, under~\Cref{assumption2} and 
\begin{itemize}
    \item $\|(\Id-h\Am)^{-1}\|\leq 1$, the query complexity scales as 
    \begin{equation}
\mathcal{O}\left(\log^{8.5} \left(\frac{\nx T^4}{\|\uv_c\|^2\epsilon^2}\right)\frac{T^3}{\|\uv_c\|^2\epsilon^2}\log(1/\epsilon)\right).
\end{equation}
    \item $\|(\Id-h\Am)^{-1}\|> 1$, the query complexity scales as 
    \begin{equation}
 \mathcal{O}\left(\log^{8.5} \left(\frac{\nx T^2}{\|\uv_c\|^2\epsilon^2}\right)\frac{T e^{2\nu T}}{\|\uv_c\|^2\epsilon^2}\log(1/\epsilon)\right).
\end{equation}
\end{itemize}

\end{theorem}

\begin{proof}
By using~\Cref{lem:kappab,lem:eulerexb}, and following similar steps as in~\Cref{thm:errfor,thm:normat}, the result follows. 
%See also the Remark \ref{rem:approxexp} in the Appendix.
\end{proof}

\begin{remark}
Given the update eqn. (\ref{eq:eulerb}),  the computational complexity of classical implicit Euler method can be estimated as
\begin{equation*}
O(\nx^2M)=O(\nx^2T/h),
\end{equation*}
where, we have assumed a GMRES based matrix inversion approach. Following similar analysis as in~\Cref{sec:cmpexp}, we  can conclude that BVQPCO will scale polylogarithmically with $N$ , under~\Cref{assumption2}, compared to classical approach and thus could provide a significant computational advantage.    
\end{remark}

\section{Numerical Results}\label{sec:num}

In this section, we demonstrate the application of our BVQPCO framework to solve the heat transfer optimization problem described in ~\Cref{sec:heatqen}. The problem is discretized on a grid of size $N=8$ and $\Delta x = \frac{l}{N-1} = 1/7$. We consider the implicit time discretization given in~\Cref{sec:impED} with $\Delta t = 0.25 s$ and $M = 4$ time steps. The following values of the parameters are used: $q_0 = 50.0, k=1.0, l_{\min} = 2.0, l_{\max} = 4.0, \alpha_{\min}=0.2, \alpha_{\max}=0.3, w_1 = 10.0, w_2 = 1.0, w_3 =5.0$.   There is a gradient in the initial temperature distribution across the grid as shown in~\Cref{fig:vqls_conv} at $t=0s$. This can be achieved by choosing the rotation angles in the state preparation circuit shown in~\Cref{fig:stateprep}.

The algorithm was implemented using the PennyLane software framework for quantum computing. PennyLane's \texttt{lightning.qubit} device was used as the simulator, \texttt{autograd} interface was used as the automatic differentiation library and the adjoint method was used for gradient computations. We do not consider any device or measurement noise throughout this study. There are several choices of optimizers which can be used with VQLS as discussed in~\Cref{sec:opt}. We experimented with the Adagrad optimizer from PennyLane with step size of $0.8$ in~\Cref{sec:conv_vqls} and COBYLA optimizer from SciPy~\cite{2020SciPy-NMeth} in~\Cref{sec:conv_design}, both of which worked well in our application. The convergence criteria for VQLS was set to a maximum of $150$ iterations. The optimizer is initialized randomly with samples taken from the Beta distribution with shape parameters $\alpha=\beta=0.5$. 

The implementation of Bayesian optimization from Ax, BoTorch~\cite{balandat2020botorch} was used with Ax's Loop API. Sobol sequences are used to generate ten initial samples in the parameter space and twenty BO iterations are performed. Fixed noise Gaussian process (where noise is provided as an input) is used as the surrogate model and qNoisy EI is used as the acquisition function. While we do not consider device or measurement noise in this work, we  use qNoisy EI model to capture the error in estimation of design cost from the inner VQLS optimization loop. We input standard error to the BO model as $\sigma_{\epsilon}(l,\alpha) \approx c\sqrt{C_g(l,\alpha) \log N}$ using~\Cref{eq:errbounds1}, where $c$ is a fixed constant we selected empirically. 

\subsection{Convergence of VQLS}
\label{sec:conv_vqls}
 \begin{figure}[htbp!]
\centering
\begin{subfigure}[b]{0.48\linewidth}
        \centering
    \includegraphics[scale=0.25]{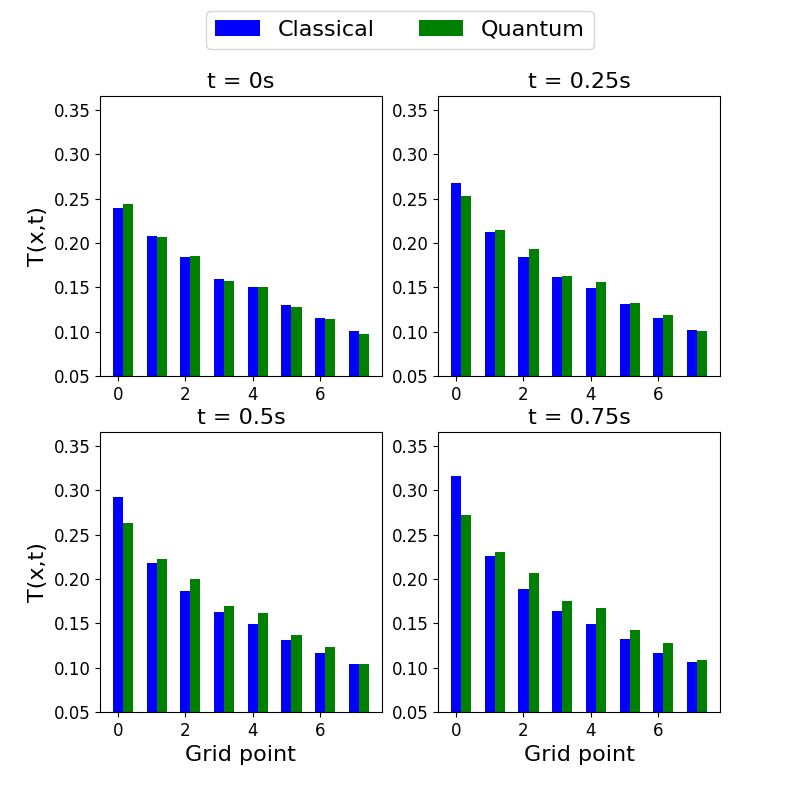}
    ~\hspace{-12pt}
    \raisebox{5pt}{
    \includegraphics[scale=0.25]{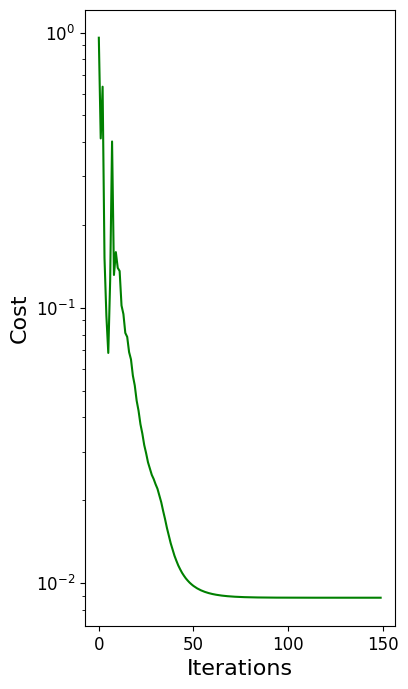}}
        \caption{$l=4.0$, $\alpha=0.2$}
    \end{subfigure}
    ~
    \begin{subfigure}[b]{0.48\linewidth}
        \centering
        \includegraphics[scale=0.25]{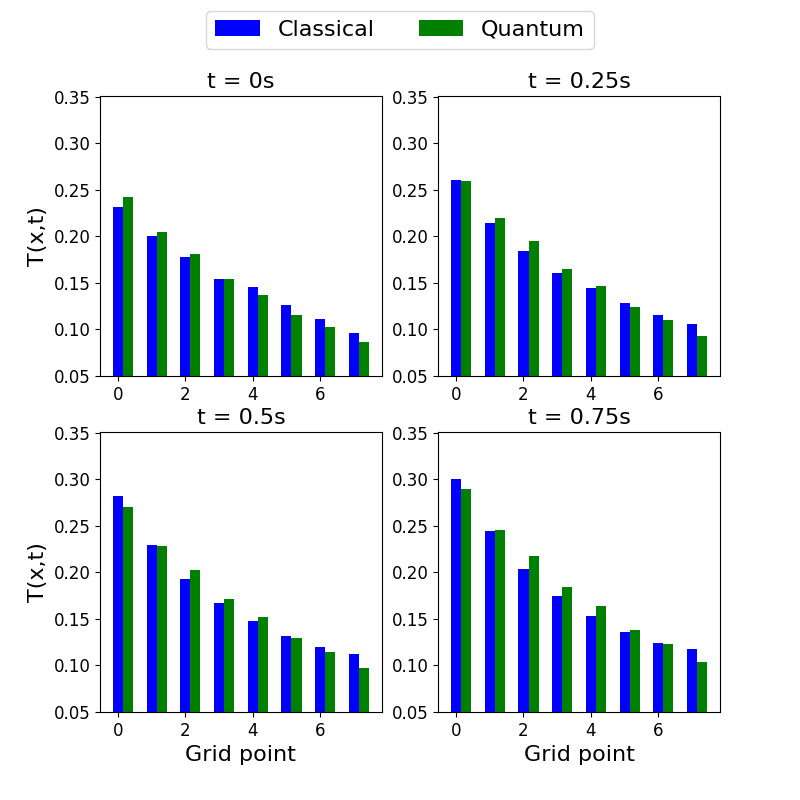}
        ~\hspace{-12pt}
         \raisebox{5pt}{\includegraphics[scale=0.25]{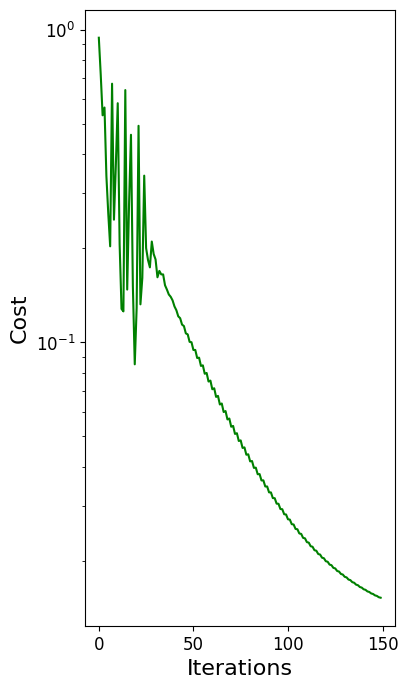}}
        \caption{$l=2.0$, $\alpha=0.3$}
    \end{subfigure}
    \caption{Convergence of VQLS for two different values of the design parameters $l, \alpha$. The global cost converges to less than $2 \times 10^{-2}$ within a maximum of 150 iterations.}
    \label{fig:vqls_conv}
\end{figure}

\begin{figure}[htbp!]
\centering
\includegraphics[scale=0.6]{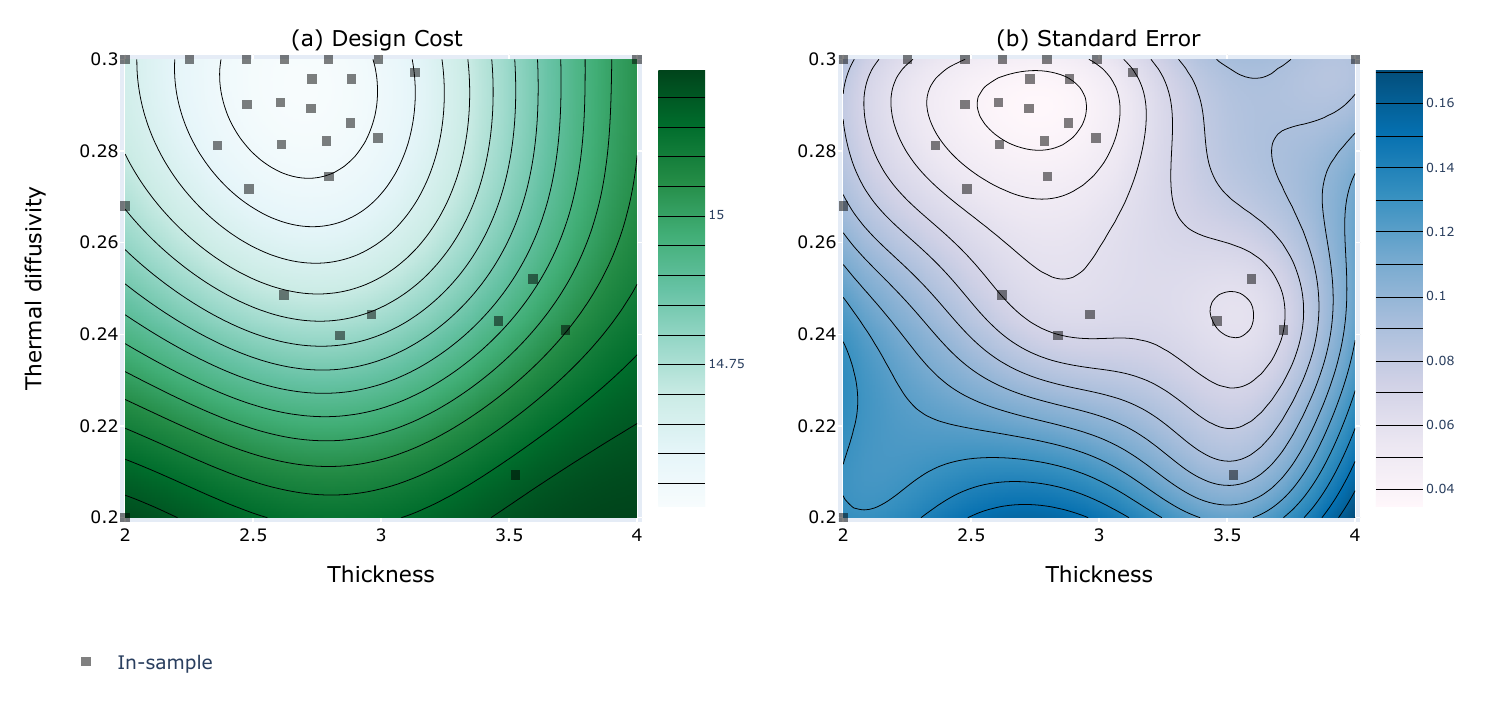}

\includegraphics[scale=0.6]{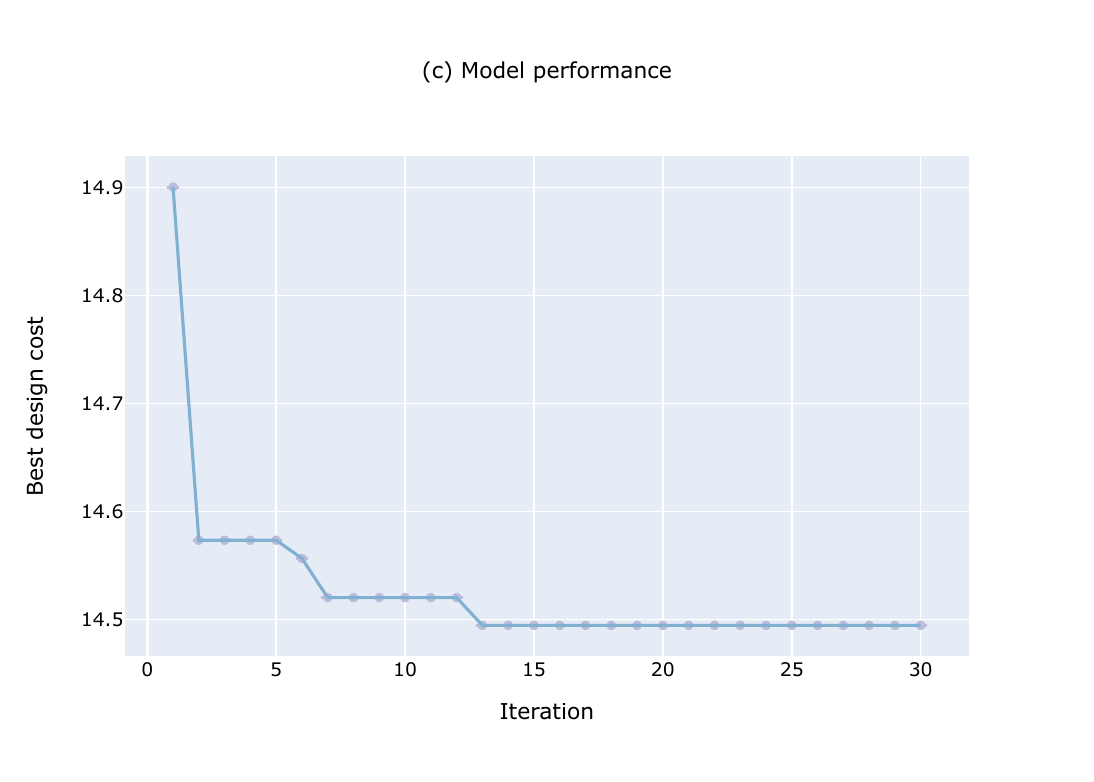}
\caption{(a) Contour plot illustrating the solution space of the design optimization problem. (b) Contour plot of the standard error in the design cost. c) Performance of the BO model with iteration count.}
\label{fig:bo_conv}
\end{figure}

First, we study the convergence of VQLS for different values of the design parameters $l, \alpha$. The unitary basis in the LCU decomposition of the $\Tilde{\Am}_{b1}$ and $\Tilde{\Am}_{b2}$  are computed once and the non-zero coefficients are updated using~\Cref{eq:Aim} for different values of the design parameters. The convergence of VQLS for two set of parameters values is shown in~\Cref{fig:vqls_conv}. The quantum solution is compared with the classical solution obtained by solving the linear system~(\ref{eq:linsysimp}). Qualitative comparison of the two solutions indicate that VQLS produces reasonable solutions to the problem. In addition, the Adagrad optimizer converges asymptotically to a global cost of less than $2\times 10^{-2}$. 

\subsection{Convergence of the BVQPCO framework}
\label{sec:conv_design}

Next, we study the performance of the BVQPCO framework on the heat transfer optimization problem (\ref{eq:optprob}).  The true optimum of the problem is found classically to be $l^* = 2.90 $, $\alpha^* = 0.278. $ Our algorithm returns the optimal solution as $l=2.73$ and $\alpha=0.285$ with less than $0.7\%$ error in the design cost.~\Cref{fig:bo_conv} shows the convergence of the BO loop. The contour plot in~\Cref{fig:bo_conv} (a) shows the design cost landscape in the parameter space and the sampled points. Note that more points are concentrated around the optimal point indicating convergence. The standard error shown in~\Cref{fig:bo_conv} (b) is a combination of the error input to the noisy EI acquisition function and the confidence of the model in predicting the design cost at a given point in the parameter space. Finally, ~\Cref{fig:bo_conv} (c) shows the performance of the model as a function of the iteration count. The best design cost found until a given iteration is plotted on the y-axis. These results demonstrate the overall accuracy and convergence of our BVQPCO framework.

\section{Conclusion}\label{sec:conc}
In this paper we presented BVQPCO, a novel variational quantum framework for PDE constrained design optimization problems. The proposed framework utilizes the VQLS algorithm and a black box optimizer nested in bi-level optimization structure. We presented a detailed computational error and complexity analysis to establish potential benefits of our framework over classical techniques. We demonstrated the framework on a heat transfer optimization problem and presented simulation results using Bayesian optimization as the black box optimizer. The analysis and results demonstrate the correctness of our framework for solving simulation-based design optimization problems. To the best of our knowledge, this is the first work of its kind in designing a bi-level variational quantum framework for solving PDE constrained design optimization problems.  

Future work will involve studying and mitigating the effect of device/ measurement noise and making the BVQPCO framework robust. It will also be worthwhile to study the scalability of the framework by applying it to larger problem sizes and implementing on quantum hardware. One of the bottlenecks in the BVQPCO framework is the computation of VQLS cost function which scales polynomially with the number of LCU terms. We can employ a more efficient tensor product decomposition~\cite{gnanasekaran2024efficient}, which exploits the underlying structure and sparsity of matrices such as arising from PDE discretizations, to provide an exponential speed-up in the VQLS cost evaluation. An adaptive multigrid strategy to accelerate convergence of variational quantum algorithms was recently proposed in~\cite{pool2024nonlinear}. Such techniques can be adapted in the context of VQLS to further improve our framework and scale to larger problem sizes. By leveraging Carleman linearization based embedding approach~\cite{liu2021efficient,surana2024efficient} we also plan to extend our framework for nonlinear PDE constrained optimization problems. Finally, adapting the proposed framework for fault tolerant quantum computing setting is also an avenue for future research.
\section{Acknowledgments}
This research was developed with funding from the Defense Advanced Research Projects Agency (DARPA). The views, opinions, and/or findings expressed are those of the author(s) and should not be interpreted as representing the official views or policies of the Department of Defense or the U.S. Government.

\bibliographystyle{unsrt}
\bibliography{references}

\appendix
\section{Proofs for~\NoCaseChange{\Cref{sec:erranalysis}}}

\subsection{Proof of~\Cref{lem:kappa}}
\label{sec:proof_lem3}
\hchoice*
\begin{proof}
Let $\overline{\Am}=\Id+h\Am$, and so
\begin{equation*}
\|\overline{\Am}\|_2\leq 1+h\nu,
\end{equation*}
where, $\nu=\|\Am\|_2$.  Using, a sequence of triangular inequalities and the submultiplicavity of $\|\cdot\|_2$, one can easily deduce
\begin{equation}\label{eq:normAF}
\|\tilde{\Am}_f\|_2\leq \|\Id\|_2+\|\overline{\Am}\|_2\leq 2+h\nu.
\end{equation}
It is easy to show that the inverse of $\tilde{\Am}$ is composed of following blocks
\begin{equation*}
(\tilde{\Am}_f^{-1})_{ij}=\begin{cases}
\Id & i=j\\
\mathbf{0} & j>i\\
(\overline{\Am})^{i-j} & j<i.
\end{cases}
\end{equation*}
Consequently, again applying a sequence of triangular inequalities and the submultiplicavity property leads to
\begin{equation*}
\|\tilde{\Am}_f^{-1}\|_2\leq \|\Id\|_2+\|\overline{\Am}\|_2+\|\overline{\Am}^2\|_2+\cdots \|\overline{\Am}^{\tsteps-1}\|_2\leq \frac{(1+h\nu)^{\tsteps}-1}{h\nu}.
\end{equation*}
Thus, the condition number of $\tilde{\Am}$ can be bounded as
\begin{eqnarray*}
\kappa(\tilde{\Am}_f)&=&\|\tilde{\Am}_f\|_2\|\tilde{\Am}_f^{-1}\|_2\leq (2+h\nu)\frac{(1+h\nu)^{\tsteps}-1}{h\nu}\\
&\leq& \frac{(2+h\nu)}{h\nu} (1+h\nu)e^{\nu T} \quad \bigg(\because \Big(1+\frac{a}{n}\Big)^n\leq e^a\bigg)\\
% &\leq & e^{\nu T}e^{h\nu}\frac{(2+h\nu)}{h\nu}
&\leq& \frac{12e^{\nu T}}{h\nu},
\end{eqnarray*}
where, we have used $T=(\tsteps-1) h$ and $h\nu\leq 2$.
\end{proof}

\subsection{Proof of~\Cref{lem:eulerex}}
\label{sec:proof_lem4}

\eulerex*
\begin{proof}
For the forward Euler scheme to be stable, the step size should satisfy
\begin{equation}
h\leq \frac{2}{\rho(\Am)}.\label{eq:hstabcond}
\end{equation}
We choose $h\leq \frac{2}{\nu}$ where $\nu=\|\Am\|_2$ which satisfies the condition (\ref{eq:hstabcond}). Furthermore, under~\Cref{assumption1} the forward Euler scheme satisfies following global error bound~\cite{ascher1998computer}
\begin{equation*}
\|\uv((k-1)h)-\uvh_f^k\|\leq Ch, \quad k=1,\dots,\tsteps,
\end{equation*}
where, $C>0$  is a constant. It then follows that
\begin{eqnarray}
\|\uv_c-\tilde{\uv}_f\|^2&=& \sum_{k=1}^{\tsteps} \|\uv((k-1)h)-\uvh_f^k\|^2\leq (Ch)^2(\tsteps-1),\notag
\end{eqnarray}
which implies
\begin{eqnarray}
\|\uv_c-\tilde{\uv}_f\|&\leq & Ch(\tsteps-1)^{1/2}\leq C h^{1/2} T^{1/2}.
\end{eqnarray}
Thus, for any given $\epsilon^\prime$ one can always choose
\begin{equation*}
h<\min\Big\{\frac{(\epsilon^\prime)^2}{C^2T},\frac{2}{\nu}\Big\},
\end{equation*}
such that
\begin{equation*}
\|\uv_c-\tilde{\uv}_f\|\leq \epsilon^\prime,
\end{equation*}
and Euler scheme is stable. Finally, note that
\begin{eqnarray}
\bigg\|\frac{\uv_c}{\|\uv_c\|}-\frac{\tilde{\uv}_f}{\|\tilde{\uv}_f\|}\bigg\|&\leq& \bigg\|\frac{\uv_c}{\|\uv_c\|}-\frac{\tilde{\uv}_f}{\|\uv_c\|}+\frac{\tilde{\uv}_f}{\|\uv_c\|}-\frac{\tilde{\uv}_f}{\|\tilde{\uv}_f\|}\bigg\|,\notag\\
&\leq & \frac{\|\uv_c-\tilde{\uv}_f\|}{\|\uv_c\|}+\frac{|\|\uv_c\|-\|\tilde{\uv}_f\||}{\|\uv_c\|}\leq \frac{2\|\uv_c-\tilde{\uv}_f\|}{\|\uv_c\|}\leq \frac{2\epsilon^\prime}{\|\uv_c\|},
\end{eqnarray}
where, we have use the fact
\begin{equation*}
\frac{|\|\uv_c\|-\|\tilde{\uv}_f\||}{\|\uv_c\|}\leq \frac{\|\uv_c-\tilde{\uv}_f\|}{\|\uv_c\|}.
\end{equation*}
Thus, for given $\epsilon$, by choosing $\epsilon^\prime=\|\uv_c\|\epsilon/2$ and determining $h$ as described above, we can conclude
\begin{equation*}
\bigg\|\frac{\uv_c}{\|\uv_c\|}-\frac{\tilde{\uv}_f}{\|\tilde{\uv}_f\|}\bigg\|\leq \epsilon,
\end{equation*}
as required.
\end{proof}

\subsection{Proof of~\Cref{thm:normat}}
\label{sec:proof_thm2}
\normat*
\begin{proof}
Given $\Am$ is a normal matrix i.e. $\Am^*\Am=\Am\Am^*$, it can always be decomposed into the form $\Am=\Vm^{T}\Lambda\Vm$ where $\Vm$ is an unitary matrix and $\Lambda$ is a diagonal matrix with diagonal entries as the eigenvalues of $\Am$. Then
\begin{eqnarray*}
  \|\Id+h\Am\|_2 &=& \|\Vm^{T}(\Id+h\Lambda)\Vm\|_2\leq \rho(\Id+h\Lambda). 
\end{eqnarray*}
We choose $h$ such that $\rho(\Id+h\Lambda) \leq  1$, or
\begin{eqnarray*}
|1+h\lambda_j|^2&\leq & 1,
\end{eqnarray*}
for all $\lambda_j$ eigenvalues of $\Am$. Let, $\lambda_j=-\delta_j+i\omega_j$ where $\delta_j> 0$ (or $\delta_j \geq 0$ if $\omega_j=0$) by Assumption \ref{assumption1}, then above condition can be met by  choosing
\begin{eqnarray*}
h\leq \min_{j}\frac{2\delta_j}{|\lambda_j|^2},
\end{eqnarray*}
Thus, for
\begin{eqnarray}
h\leq h_0\equiv\min_{j}\left\{\frac{2\delta_j}{|\lambda_j|^2},\frac{2}{\max_{k}{|\lambda_k|}}\right\}, \label{eq:h0}
\end{eqnarray}
following holds
\begin{equation}\label{eq:diffhbound}
\|\Id+h\Am\|_2\leq 1,
\end{equation}
and the condition for stability of forward Euler is also satisfied.

\paragraph{Euler error:} We next revisit error bounds for forward Euler under the assumption (\ref{eq:diffhbound}). By Taylor expansion
\begin{equation*}
\uv((k+1)h)=\uv(kh)+h\Am\uv(kh)+h\bv(kh)+\tauv_k,
\end{equation*}
where, $\tauv_k$ is the local truncation error which can be expressed as
\begin{equation}
\|\tauv_k\|\leq\frac{M_k}{2}h^2,\label{eq:taubound}
\end{equation}
under the Assumptions \ref{assumption1}, where $M_k>0$ is a constant. Thus
\begin{equation*}
\uv(kh)-\uvh_f^{k+1}=\uv((k-1)h)+h\Am\uv((k-1)h)+h\bv((k-1)h)+\tauv_k-(\uvh_f^k+h\Am\uvh_f^k+h\bv((k-1)h)),
\end{equation*}
which implies the norm of the error $e_{k}=\|\uv((k-1)h)-\uvh_f^{k}\|$ satisfies
\begin{equation}
e_{k+1}\leq \|(\Id+h\Am)\|e_k+\|\tauv_k\|.
\end{equation}
Iterating on $k$ we get
\begin{equation}
e_{k+1}\leq \|(\Id+h\Am)\|^{k}e_1+(\|(\Id+h\Am)\|^{k-1}\|\tauv_1\|+\|(\Id+h\Am)\|^{k-2}\|\tauv_2\|+\cdots \|\tauv_k\|).
\end{equation}
Since $e_1=0$ for Euler method, using the condition (\ref{eq:diffhbound}) and (\ref{eq:taubound}) for some constant $C>0$, we conclude that
\begin{equation*}
\|\uv((k-1)h)-\uvh_f^{k}\|=e_{k}\leq C (k-1)h^2.
\end{equation*}
It then follows that
\begin{eqnarray}
\|\uv_c-\tilde{\uv}_f\|^2&=& \sum_{k=1}^{\tsteps} \|\uv((k-1)h)-\uvh_f^k\|^2\leq (Ch^2)^2\sum_{k=1}^{\tsteps}(k-1)^2,\notag
\end{eqnarray}
which implies
\begin{eqnarray}
\|\uv_c-\tilde{\uv}_f\|&\leq & Ch^2(\tsteps-1)^{3/2}\leq C h^{1/2} T^{3/2}.
\end{eqnarray}
Thus, for any given $\epsilon^\prime$ one can choose 
\begin{equation}
h<\min \left\{\frac{(\epsilon^\prime)^2}{C^2T^3},h_0\right\},
\end{equation}
where $h_0$ is as defined in (\ref{eq:h0}),  so that
\begin{equation*}
\|\uv_c-\tilde{\uv}_f\|\leq \epsilon^\prime.
\end{equation*}
Finally, for given $\epsilon$, by choosing $\epsilon^\prime=\|\uv_c\|\epsilon/2$, selecting $h$ as described above, and following the steps in the Lemma \ref{lem:eulerex}, we conclude
\begin{equation*}
\left\|\frac{\uv_c}{\|\uv_c\|}-\frac{\tilde{\uv}_f}{\|\tilde{\uv}_f\|}\right\|\leq \epsilon.
\end{equation*}

\paragraph{Condition number:} Under the condition (\ref{eq:diffhbound}), following similar steps as in the proof of the Lemma \ref{lem:kappa}, one can deduce
\begin{equation*}
\|\tilde{\Am}_f\|_2\leq \|\Id\|_2+\|\overline{\Am}\|_2\leq 2,
\end{equation*}
and
\begin{equation*}
\|\tilde{\Am}_f^{-1}\|_2\leq \|\Id\|_2+\|\overline{\Am}\|_2+\|\overline{\Am}^2\|_2+\cdots \|\overline{\Am}^{\tsteps-1}\|_2\leq \tsteps,
\end{equation*}
where we used the fact that $\|\overline{\Am}^k\|_2\leq (\|\overline{\Am}\|_2)^k\leq 1$ for $k=1,2,\dots, M-1$. Thus, the condition number of $\tilde{\Am}_f$ can be bounded as
\begin{eqnarray*}
\kappa(\tilde{\Am}_f)&=&\|\tilde{\Am}_f\|_2\|\tilde{\Am}_f^{-1}\|_2\leq 2 \tsteps.
\end{eqnarray*}

Finally, choose 
\begin{equation}
h<\min \left\{ \frac{(\epsilon\|\uv_c\|)^2}{16C^2T^3},h_0\right\},
\end{equation}
and
\begin{equation}\label{eq:gamma1}
\gamma=\frac{1-(1-\frac{\epsilon^2}{8})^2}{2\,(2M)^2\log(\nx\tsteps)}<1,
\end{equation}
and follow steps similar to the proof of~\Cref{thm:errfor} to conclude that 
\begin{equation}
\| |\uv_c\ra-|\uvo_f\ra\|\leq \epsilon.
\end{equation}

Furthermore, under~\Cref{assumption2} the query complexity can be determined as follows
\begin{eqnarray*}
 && \mathcal{O}\left(\log^{8.5} (\nx\tsteps)\kappa\log(1/\epsilon^\prime)\right) = \mathcal{O}\left(\log^{8.5} (\nx\, T/h)2(T/h)\log(1/\epsilon)\right)\\
  &=&  \mathcal{O}\left(\log^{8.5} \left(\nx T\left(\frac{1}{h_0}+\frac{4C^2T^3}{\|\uv_c\|^2\epsilon^2}\right)\right)\left(\frac{1}{h_0}+\frac{4C^2T^3}{\|\uv_c\|^2\epsilon^2}\right)\log(1/\epsilon)\right)\\
  &=& \mathcal{O}\left(\log^{8.5} \left(\frac{\nx T^4}{\|\uv_c\|^2\epsilon^2}\right)\frac{T^3}{\|\uv_c\|^2\epsilon^2}\log(1/\epsilon)\right),
\end{eqnarray*}
where, we have assumed $\mathcal{O}(h_0)=1$.

\end{proof}

\subsection{Proof of~\Cref{lem:kappab} }
\label{sec:proof_lem5}
\kappab*

Let $\overline{\Am}=\Id-h\Am$, and so
\begin{equation*}
1\leq \|\overline{\Am}\|_2\leq 1+h\nu,
\end{equation*}
where, $\nu=\|\Am\|_2$ and we have used the fact that $\|\overline{\Am}\|_2\geq \rho(\Id-h\Am)\geq 1$ based on~\Cref{assumption1}.  Using, a sequence of triangular inequalities and the submultiplicavity of $\|\cdot\|_2$, one can easily deduce
\begin{equation*}
\|\tilde{\Am}_b\|_2\leq \|\Id\|_2+\max\{\|I\|_2,\|\overline{\Am}\|_2\}\leq 1+\|\overline{\Am}\|_2\leq 2+h\nu.
\end{equation*}

Next note that $(\Id-h\Am)$ is invertible since its eigenvalues $\frac{1}{1-h\lambda}$ are non-zero under~\Cref{assumption1}, condition A1. Using the matrix inversion lemma, we get
\begin{equation*}
(\Id-h\Am)^{-1}=\sum_{k=0}^{\infty}(hA)^k,
\end{equation*}
and implies
\begin{equation}\label{eq:hboundback}
\|\overline{\Am}^{-1}\|=\|(\Id-h\Am)^{-1}\|\leq \frac{1}{1-h\nu}\leq 1 + 2\, h \nu,
\end{equation}
assuming $h$ is chosen such that $h\nu<0.5$.

It can be shown that the inverse of $\tilde{\Am}_b$ is composed of following blocks
\begin{equation*}
(\tilde{\Am}_b^{-1})_{ij}=\begin{cases}
(\overline{\Am}^{-1})^{i-1} & j=1\\
\overline{\Am}^{-1} & i=j,j\neq 1\\
(\overline{\Am}^{-1})^{(i-j)+1} & i>j, j\neq 1
\end{cases}.
\end{equation*}
Consequently, applying a sequence of triangular inequalities and using the submultiplicative property leads to
\begin{equation*}
\|\tilde{\Am}_b^{-1}\|_2\leq \max\{\|I\|_2,\|\overline{\Am}^{-1}\|_2\}+\max\{\|\overline{\Am}^{-1}\|_2,\|\overline{\Am}^{-2}\|_2\}+\cdots \max\{\|\overline{\Am}^{-(M-2)}\|_2,\|\overline{\Am}^{-(M-1)}\|_2\}+\|\overline{\Am}^{-(M-1)}\|_2.
\end{equation*}

We consider two cases.
\paragraph{Case 1:} $\|\overline{\Am}^{-1}\|\leq 1$, which implies $\|\overline{\Am}^{-k}\|\leq 1$
\begin{equation*}
\|\tilde{\Am}_b^{-1}\|_2\leq \tsteps.
\end{equation*}
Thus, the condition number of $\tilde{\Am}_b$ can be bounded as
\begin{eqnarray}
\kappa(\tilde{\Am}_b)&=&\|\tilde{\Am}_b\|_2\|\tilde{\Am}_b^{-1}\|_2\leq (2+h\nu)\tsteps\leq 2.5(T/h+1), \label{eq:kappaexpeuler1}
\end{eqnarray}
where, we have used $T=(\tsteps-1)h$ and assumed $h\nu<0.5$.

\paragraph{Case 2:} $\|\overline{\Am}^{-1}\|>1$, which implies 
\begin{eqnarray*}
\|\tilde{\Am}_b^{-1}\|_2&\leq &\|\overline{\Am}^{-1}\|_2+\|\overline{\Am}^{-2}\|+\cdots \|\overline{\Am}^{-(M-1)}\|+\|\overline{\Am}^{-M}\|_2\notag\\
&\leq & (1+2h\nu) \frac{(1+2h\nu)^M}{2h\nu} \leq \frac{(1+2h\nu)^2}{2h\nu}e^{2\nu T} \\
&\leq& \frac{2}{h\nu}e^{2\nu T}.
\end{eqnarray*}
%where, we have used $(1+x)^M\leq e^x$ and $h\tsteps=T$. 
Thus, condition number of $\tilde{\Am}_b$ can be bounded as
\begin{eqnarray}
\kappa(\tilde{\Am}_b)\leq (2+h\nu)\frac{2}{h\nu}e^{2\nu T}\leq \frac{5}{h\nu}e^{2\nu T}.\label{eq:kappaexpeuler2}
\end{eqnarray}

\subsection{Proof of~\Cref{lem:eulerexb}}
\label{sec:proof_lem6}
\eulerexb*
%\begin{proof}
Under the Assumptions \ref{assumption1}, the backward Euler scheme is stable. For backward Euler, the error equation is 
\begin{equation*}
\uv(kh)-\uvh_b^{k+1}=\uv((k-1)h)+h\Am\uv(kh)+h\bv(kh)+\tauv_k-(\uvh_b^k+h\Am\uvh_b^{k+1}+h\bv(kh)),
\end{equation*}
which implies the norm of the error  $e_k=\|\uv((k-1)h)-\uvh_b^{k}\|$ satisfies
\begin{equation}
e_{k+1}\leq \|(\Id-h\Am)^{-1}\|\left(e_k+\|\tauv_k\|\right),
\end{equation}
where, as before the local truncation error $\tauv_k$ satisfies $\|\tauv_k\|\leq \frac{M_kh^2}{2}$ for $M_k\geq 0$.  

As in Lemma \ref{lem:kappab}, we consider two cases under the assumption $h\nu<0.5$.
\paragraph{Case I:} $\|(\Id-h\Am)^{-1}\|\leq 1$ under which backward scheme satisfies following global error bound
\begin{equation*}
\|\uv((k-1)h)-\uvh_b^k\|\leq C_1 h^2(k-1),
\end{equation*}
where, $C_1>0$  is a constant. Thus, following similar steps as in the proof of Theorem \ref{thm:normat}, for any given $\epsilon$ one can choose 
\begin{equation}
h<\min\left\{\frac{(\epsilon \|\uv_c\|)^2}{4C_1^2T^3},\frac{0.5}{\nu}\right\},\label{eq:hboundexp1}
\end{equation}
such that
\begin{equation}
\left\|\frac{\uv_c}{\|\uv_c\|}-\frac{\tilde{\uv}_b}{\|\tilde{\uv}_b\|}\right\|\leq \epsilon^\prime.\label{eq:errEb}
\end{equation}

\paragraph{Case II:} $\|(\Id-h\Am)^{-1}\|>1$ under which backward scheme satisfies following global error bound
\begin{equation*}
\|\uv((k-1)h)-\uvh_b^k\|\leq C_2 h,
\end{equation*}
where, $C_2>0$  is a constant. Thus, following similar steps as in the Lemma \ref{lem:eulerex}, for any given $\epsilon$ one can choose 
\begin{equation}
h<\min\left\{\frac{(\epsilon \|\uv_c\|)^2}{4C_2^2T},\frac{0.5}{\nu}\right\}\label{eq:hboundexp2}
\end{equation}
such that (\ref{eq:errEb}) holds.

\end{document}

%% file: preamble.tex
\usepackage[utf8]{inputenc}
\usepackage{amsmath}
\usepackage[nameinlink]{cleveref}
\crefname{theorem}{thm.}{thm.}
\crefname{section}{sec.}{sec.}
\crefname{corollary}{cor.}{cor.}
\crefname{lemma}{lem.}{lem.}
\Crefname{theorem}{Thm.}{Thm.}
\Crefname{section}{Sec.}{Sec.}
\Crefname{figure}{Fig.}{Fig.}
\Crefname{equation}{Eqn.}{Eqn.}
\Crefname{assumption}{Assumption}{Assumption}
\usepackage{adjustbox}
\usepackage{graphicx}
\usepackage{tikz}
\usepackage{tikzscale}
\usepackage{quantikz}
\usepackage{pgfplots}
\usepackage{pgf-pie}
\pgfplotsset{compat=1.15}
\usepackage{xcolor}
\definecolor{classical}{HTML}{DBD3F5}
\definecolor{quantum}{HTML}{FBE7E6}
\usetikzlibrary{arrows.meta,decorations.markings}
\tikzset{
     heat/.style={brown,thick,-Latex, decoration={coil,aspect=0},decorate}
  }

\usepackage{wrapfig}
\usepackage{lscape}
\usepackage{rotating}

\usepackage{algorithm}
\usepackage{algcompatible}
\usepackage[utf8]{inputenc}
\usepackage{amsthm}

\usepackage{mathrsfs}
\usepackage[normalem]{ulem}
\usepackage{array}
\usepackage{subcaption}

\usepackage{multicol}
\usepackage{cuted} 
\usepackage{booktabs}
\usepackage{fancyhdr}

\usepackage{thm-restate}

\newtheorem{theorem}{Theorem}
\newtheorem{lemma}{Lemma}
\newtheorem{assumption}{Assumption}

\newtheorem{remark}{Remark}

\allowdisplaybreaks

\newcommand{\Tv}{\mathbf{u}}
\newcommand{\ev}{\mathbf{e}}
\newcommand{\bv}{\mathbf{b}}
\newcommand{\xv}{\mathbf{x}}
\newcommand{\yv}{\mathbf{y}}
\newcommand{\psiv}{\mathbf{\psi}}
\newcommand{\phiv}{\mathbf{\phi}}
\newcommand{\pv}{\mathbf{p}}
\newcommand{\Zev}{\mathbf{0}}

\newcommand{\Am}{\mathbf{A}}
\newcommand{\Xm}{\mathbf{X}}
\newcommand{\Ym}{\mathbf{Y}}
\newcommand{\Zm}{\mathbf{Z}}
\newcommand{\Um}{\mathbf{U}}

\newcommand{\In}{\mathbf{I}}
\newcommand{\Id}{\mathbf{I}}
\newcommand{\Pm}{\mathbf{P}}
\newcommand{\Vm}{\mathbf{V}}

\newcommand{\epsv}{\mathbf{\epsilon}}

\newcommand{\thetav}{\mathbf{\theta}}
\newcommand{\ra}{\rangle}
\newcommand{\la}{\langle}

\newcommand{\uv}{\mathbf{u}}
\newcommand{\uvh}{\hat{\mathbf{u}}}
\newcommand{\uvo}{\overline{\mathbf{u}}}
\newcommand{\nx}{N}
\newcommand{\tsteps}{M}
\newcommand{\Rr}{\mathbb{R}}
\newcommand{\Cr}{\mathbb{C}}
\newcommand{\Nr}{\mathbb{N}}
\newcommand{\tra}{\prime}
\newcommand{\tauv}{\mathbf{\tau}}

\newcommand{\gv}{\mathbf{g}}
\newcommand{\po}{\mathbf{p}}
\newcommand{\Po}{\mathbf{X}}

\newcommand{\bm}{\mathbf{\beta}_m}
\newcommand{\bk}{\mathbf{\beta}_k}
\newcommand{\bo}{\mathbf{\beta}}
\newcommand{\aq}{a}

\newcommand{\fo}{f}
\newcommand{\fov}{\mathbf{f}}
\newcommand{\Bm}{\mathbf{\Lambda}}
\newcommand{\sv}{\mathbf{s}}

\usepackage{listings}

\definecolor{backcolour}{rgb}{0.95,0.95,0.92}
\definecolor{codegreen}{rgb}{0,0.6,0}

\lstdefinestyle{myStyle}{
    backgroundcolor=\color{backcolour},   
    commentstyle=\color{codegreen},
    basicstyle=\ttfamily\normalsize,
    breakatwhitespace=false,         
    breaklines=true,                 
    keepspaces=true,    
    numbersep=5pt,                  
    showspaces=false,                
    showstringspaces=false,
    showtabs=false,                  
    tabsize=2,
}

\usepackage{color}
\definecolor{deepblue}{rgb}{0,0,0.5}
\definecolor{deepred}{rgb}{0.6,0,0}
\definecolor{deepgreen}{rgb}{0,0.5,0}

\newcommand\pythonstyle{\lstset{
language=Python,
basicstyle=\ttfamily,
morekeywords={self},              % Add keywords here
keywordstyle=\ttfamily\color{deepblue},
emph={MyClass,__init__},          % Custom highlighting
emphstyle=\ttfamily\color{deepred},    % Custom highlighting style
stringstyle=\color{deepgreen},
frame=tb,                         % Any extra options here
showstringspaces=false
}}

\lstnewenvironment{python}[1][]
{
\pythonstyle
\lstset{#1}
}
{}

\usetikzlibrary{calc,trees,positioning,arrows,chains,shapes.geometric,%
    decorations.pathreplacing,decorations.pathmorphing,shapes,%
    matrix,shapes.symbols}

\usetikzlibrary{positioning,fit,shapes.geometric,backgrounds}
    
\tikzset{
>=stealth',
  punktchain/.style={
    rectangle, 
    rounded corners, 
    draw=black, very thick,
    text width=11em, 
    minimum height=8em, 
    text centered, 
    on chain},
  line/.style={draw, thick, <-},
   punktchainVert/.style={
    rectangle, 
    rounded corners, 
    draw=black, very thick,
    text width=20em, 
    minimum height=8em, 
    text centered, 
    on chain},
  line/.style={draw, thick, <-}
  element/.style={
    tape,
    top color=white,
    bottom color=blue!50!black!60!,
    minimum width=8em,
    draw=blue!40!black!90, very thick,
    text width=10em, 
    minimum height=3.5em, 
    text centered, 
    on chain},
  every join/.style={->, thick,shorten >=1pt},
  decoration={brace},
  tuborg/.style={decorate},
  tubnode/.style={midway, right=2pt},
}

\usetikzlibrary{
  arrows.meta,                        % for arrow tips
  backgrounds,                        % for background layer
  ext.paths.ortho,                    % for ortho paths
  ext.positioning-plus,               % for 
  ext.node-families.shapes.geometric, % loads ext.node-families and
  calc}                               % for ($<calculations>$)
\tikzset{
  basic box/.style={
    shape=rectangle, rounded corners, align=center, draw=#1, fill=#1!25},
  header node/.style={
    node family/width=header nodes,
    text depth=+.3ex, fill=white, draw},
  header/.style={%
    inner ysep=+1.5em,
    append after command={
      \pgfextra{\let\TikZlastnode\tikzlastnode}
      node [header node] (header-\TikZlastnode) at (\TikZlastnode.north) {#1}
      node [span=(\TikZlastnode)(header-\TikZlastnode)]
           at (fit bounding box) (h-\TikZlastnode) {}
    }
  },
  fat blue line/.style={ultra thick, blue},
  fat black line/.style={thick, black}
}

%% file: heat_flow.tikz
\begin{tikzpicture}
    % Soild
    \draw(0,-0.5) -- (2, -0.5);
    \draw (0, -0.9) -- (2,-0.9);
    \draw (0, -1.4) -- (2,-1.4);
    % Stripes
    \foreach \x in {0.2,0.4, 0.6, 0.8, 1.0, 1.2, 1.4, 1.6, 1.8, 2.0} {
      \draw (\x, -0.5)--(\x -0.2, -0.9);
    }
    % Heat flow
    \foreach \x in {0.2,0.6,1.0,1.4,1.8} {
      \draw[heat](\x,0.3)--++(0, -0.8);
    }
    % height
    % \draw[<->] (2.2, -0.5) -- (2.2, -0.8);
    % \node at (2.4, -0.65) {\small \textit{l}} ;
    \node at (2.4, 0) {\small$q(t)$} ;
    \node at (2.5, -0.5) {\small$x=0$} ;
    \node at (2.5, -0.9) {\small$x=l$} ;
\end{tikzpicture}

%% file: Figure2.tikz
\scalebox{0.75}{\begin{tikzpicture}
  [node distance=3cm,
  start chain=going right,]
     \node[punktchain] (intro) {ODE\\ %\vspace{-0.7cm}
     \begin{eqnarray*}
\dot{\uv}&=&\Am \uv+\bv(t),\\
\uv(0)&=&\uv_0,
\end{eqnarray*}
\begin{equation*}
\uv_c=\left(\begin{array}{cc}
     \uv(0)  \\
     \uv(h)   \\
     \vdots \\
     \uv(h(M-1))
\end{array}\right).
\end{equation*}};
     \node[punktchain] (euler)      {Difference Equation\\ %\vspace{-0.7cm}
     \begin{align*}
\uvh_f^{k+1}&=(\Id+h\Am)\uvh_f^k +h\bv^k,\\
\uvh_f^1&=\uv(0),\\
\tilde{\uv}_f&=\left(\begin{array}{c}
         \uvh_f^1 \\
         \uvh_f^2 \\
         \vdots \\
         \uvh_f^\tsteps \\
       \end{array}\right).
\end{align*}
};
     \node[punktchain, text width= 8em, shift=(left:2em)] (vqls)      {VQLS\\ %\vspace{-0.7cm}
     \begin{equation*}
\tilde{\Am}_f\Vm(\thetav)|\Zev\ra=|\tilde{\bv}\ra,
\end{equation*}
$|\uvo_f\ra=\Vm(\thetav_*)|\Zev\ra$.};

\node (foo) [left=2.1cm of intro] {PDE};
\node (end) [right=1.8cm of vqls] {};

\draw [thick, ->]  (foo)  to node[above] {Spatial} (intro);
\draw [thick, ->]  (foo)  to node[below] {Discretization} (intro);
\draw [thick, ->] (intro) to node[above] {Euler Discretization} (euler);
\draw [thick, ->] (euler) to node[above] {Linear System} (vqls);
\draw [thick, ->] (intro) to node[below] {$\||\uv_c\ra-|\tilde{\uv}_f\ra\|$} (euler);
\draw [thick, ->] (euler) to node[below] {$\||\tilde{\uv}_f\ra-|\uvo_f\ra\|$} (vqls);
\draw [thick, ->]  (vqls)  to node[above] {VQLS} (end);
\draw [thick, ->]  (vqls)  to node[below] {Solution} (end);
  \end{tikzpicture}}